\documentclass[11pt]{article}
\usepackage[utf8]{inputenc}
\usepackage{booktabs} 
\usepackage[noend, ruled]{algorithm2e} 
\SetAlFnt{\small}
\SetAlCapFnt{\small}
\SetAlCapNameFnt{\small}
\SetAlCapHSkip{0pt}
\IncMargin{-\parindent}

\usepackage{verbatim}
\usepackage{setspace}

\usepackage{amssymb}
\usepackage{amsmath}
\usepackage{amsthm}
\usepackage{mathtools}
\usepackage{tensor}
\usepackage{pgf}
\usepackage{pgfplots}
\usepackage{tikz}
\usepackage{graphicx}
\graphicspath{ {./images/} }
\usepackage{breqn}
\usepackage{hyperref}
\usepackage{nameref}
\usepackage[shortlabels]{enumitem}
\usepackage{thmtools}
\usepackage{thm-restate}

\usetikzlibrary{decorations}
\usetikzlibrary[decorations]
\usepgflibrary{decorations.pathmorphing}
\usepgflibrary[decorations.pathmorphing]
\usetikzlibrary{decorations.pathmorphing}
\usetikzlibrary[decorations.pathmorphing]
\usepackage{tikzit}

\tikzstyle{standard}=[fill=white, draw=black, shape=circle, scale=1]
\tikzstyle{label}=[fill=white, draw=white, scale=1]

\tikzstyle{weighted directed}=[->]
\tikzstyle{squiggly}=[->, decoration={{snake}}, tikzit draw={rgb,255: red,191; green,255; blue,0}, decorate]

\usepackage{subcaption}
\usepackage{float}
\usepackage{natbib}

\newtheorem{theorem}{Theorem}
\newtheorem{lemma}{Lemma}

\newtheorem{proposition}{Proposition}
\newtheorem{corollary}{Corollary}

\theoremstyle{remark}

\newcounter{fig}
\usepackage{todonotes}

\hypersetup{
     colorlinks=true,      		
     linkcolor=blue,        
     citecolor=magenta,     
     filecolor=green,      		
     urlcolor=orange,           	
}

\DeclareMathOperator*{\argmin}{argmin}

\newcommand{\eps}{\varepsilon}
\newcommand{\eqdef}{\overset{\text{def}}{=}}

\title{Chunking Tasks for Present-Biased Agents}
\author{Joe Halpern$^1$, Aditya Saraf$^1$}
\date{$^1$Cornell University\\[2ex]
\today}

\begin{document}

\maketitle

\begin{abstract}
    Everyone puts things off sometimes. How can we combat this
    tendency to procrastinate? A well-known technique used by
    instructors is to break up a large project into more manageable
    chunks. But how should this be done best? Here we study the
    process of chunking using the graph-theoretic model of present
    bias introduced by \citet{kleinberg2014naive}. We first analyze
    how to optimally chunk single edges within a task graph, given a
    limited number of chunks. We show that for edges on the
     shortest path, the optimal chunking makes
    initial chunks easy and later chunks progressively harder. For
    edges not on the shortest path, optimal chunking is
significantly more complex, but we provide an efficient
algorithm that chunks the edge optimally.  We then use our optimal
edge-chunking algorithm to 
    optimally chunk task graphs. 
    We show that with a linear number of chunks on each edge, the biased agent's cost can be exponentially lowered, to within a constant factor of the true cheapest path.
    Finally, we extend our model to the case where a task designer must chunk a graph for multiple types of agents simultaneously. The problem grows significantly more complex with even two types of agents, but we provide optimal graph chunking algorithms for two types.
Our work highlights
    the efficacy of chunking as a means to combat present bias. 
\end{abstract}

\section{Introduction}

Everyone puts things off sometimes. How can we combat this tendency to
procrastinate? A well-known technique used by instructors is to break
up a large project into more manageable chunks. But how should this be
done best? Here we study the process of chunking using the
graph-theoretic model of \emph{present bias} introduced by
\citet{kleinberg2014naive}. One of our main results confirms the
intuition long held by teachers: in many cases, the best way to chunk
a single task involves making the initial subtasks easy and then
getting progressively harder. We also provide algorithms that can best
``distribute'' chunks across many tasks, which could be applied in an
automated to-do list chunking app. 

Present bias is the tendency of agents to overweight costs and rewards
experienced in the current time period, which helps explain many
irrational behaviors, from procrastination to task
abandonment. \citet{kleinberg2014naive} had the crucial insight that
this diverse behavior could be captured in a single graph-theoretic
model. They represent tasks using a directed, acyclic graph $G$, with
designated start $s$ and end $t$. A path through this graph
corresponds to a plan to complete the task; each edge represents one
step of this plan. The weights on edges represent the costs of
completing that step. While the model is simple, it is
deceptively complex to analyze; it has been a popular starting point
for present bias in the CS community (see, e.g.,
\citep{gravin2016procrastination,albers2017value,oren2019principal,ma2019penalty,anagnostopoulos2020collaborative,fomin2020present}).

The goal of an agent is to complete the task while incurring the least
cost. An optimal (unbiased) agent simply computes the shortest path
and takes it. A \textit{naive} present-biased agent with bias
parameter $b > 1$ behaves as follows. At $s$, they compute their
perceived cost of each path to $t$ by scaling up the cost of the first
edge on each path by $b$. Then they take one step along this path, say
to vertex $u$, and then \textit{recompute} their perceived costs, this
time by scaling up the costs on the edges out of $u$. Notice that
the agent may plan to take some path $P$ at $s$, but then deviate 
from their plan after one step. This is because they (naively) do not take the 
future impact of their present bias into account when
planning; see \autoref{fig:branching} for an example. 

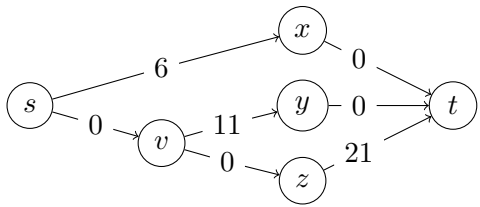
\begin{figure}[h]
    \centering
    \begin{tikzpicture}
	\begin{pgfonlayer}{nodelayer}
		\node [style=standard] (0) at (0, 0) {$s$};
		\node [style=standard] (2) at (7.25, 2) {$x$};
		\node [style=standard] (5) at (3.5, -1) {$v$};
		\node [style=label] (6) at (5.25, -1.5) {$0$};
		\node [style=standard] (7) at (11.25, 0) {$t$};
		\node [style=standard] (8) at (7.25, -2) {$z$};
		\node [style=label] (9) at (8.75, -1.25) {$21$};
		\node [style=standard] (13) at (7.25, 0) {$y$};
		\node [style=label] (14) at (3.5, 1) {$6$};
		\node [style=label] (17) at (1.75, -0.5) {$0$};
		\node [style=label] (18) at (8.75, 0) {$0$};
		\node [style=label] (19) at (8.75, 1.25) {$0$};
		\node [style=label] (20) at (5.25, -0.5) {$11$};
	\end{pgfonlayer}
	\begin{pgfonlayer}{edgelayer}
		\draw (5) to (6);
		\draw [style=weighted directed] (6) to (8);
		\draw (8) to (9);
		\draw [style=weighted directed] (9) to (7);
		\draw (0) to (14);
		\draw (2) to (19);
		\draw [style=weighted directed] (19) to (7);
		\draw (13) to (18);
		\draw [style=weighted directed] (18) to (7);
		\draw (0) to (17);
		\draw [style=weighted directed] (17) to (5);
		\draw [style=weighted directed] (14) to (2);
		\draw (5) to (20);
		\draw [style=weighted directed] (20) to (13);
	\end{pgfonlayer}
\end{tikzpicture}
    \caption{Taken from \citet{saraf2020competition}. The cheapest path is $(s, x, t)$ with total cost $6$. However, an agent with bias $b = 2$ will take path $(s, v, z, t)$, with cost $21$. Importantly, when the agent is deciding which vertex to move to from $s$, they evaluate the path starting with $x$ as having total cost $12$, while the path starting with $v$ has total cost $11$. This is because they assume they will behave optimally at $v$ by taking path $(v, y, t)$. However, they apply their bias at $v$ and deviate to the most expensive path.}
    \label{fig:branching}
\end{figure}

We extend the Kleinberg-Oren model by giving a task
designer the power to break up an edge into chunks. The agent
completes the chunks one at a time, which reduces the impact of their
present bias. We consider the chunks to be a mental feature -- the
designer does not actually check that the agent completes the task in
chunks, but instead suggests a chunking to the agent. Our model is
a good fit for many, but not necessarily all tasks. We now
highlight three families of applications 
and consider the extent to which our results apply to them.

The first family of applications are personal tasks, such as in the example given by George Akerlof of repeatedly putting off an errand until the next day \citep{akerlof1991procrastination}. In these examples, we believe that chunking can be an effective tool. Breaking even a simple task like ``mailing books'' down into smaller components like ``gather the books'', ``package the books'', and ``drive to the post office'' seems like a typical way to convince oneself to do an errand.
However, there is no real task designer here. Further, our results assume a known bias, but agents in our model are not fully aware that they have present bias. Thus, personal tasks are not the main application we consider (though our overall takeaway that chunking is valuable still applies to these tasks).

Next we consider educational examples, where students procrastinate on course work (while not planning with this in mind). Our model applies well here, as the task designer (the instructor) really does have a vested interest in ensuring that students complete the course, and do so as efficiently as possible. As mentioned before, we do not model the teacher as actively enforcing the chunks, for example with grades or deadlines. Our model is better understood as the teacher suggesting chunks to the students. We discuss further at the end of Section 2.

Finally, another application with great potential is to automatically
chunk to-do lists.
Consider an app that automatically  
takes in a user's to-do list, which could have multiple dependencies,
and suggests ways to chunk some tasks. To avoid overwhelming the user,
the app would not want to suggest too many chunks.
\footnote{Interestingly, it seems that Google's acquisition of the startup Timeful
has led to users of Gmail getting various ``nudge'' reminders,
where the nudges chosen are based in part on
research on present bias [J. Kleinberg, private communication,
  2022].}

We are not the first to consider ways of alleviating the harm caused by
present bias
(which can be quite significant---as shown by
\citet{kleinberg2014naive} and \citet{tang2017computational}, the
ratio of the optimal agent's cost to the biased agent's cost can be
exponential in the size of the graph).     
\citet{kleinberg2014naive} propose a model where a reward
is given after finishing the task, and where the agent will abandon
the task if at any point, they perceive the remaining cost to be
higher than the reward. Unlike an optimal agent, a biased agent may
abandon a task partway through; see \autoref{fig:gym} for an example.  
As a result, Kleinberg and Oren give the task designer the power to
arbitrarily delete vertices and edges, which can model deadlines.
They then investigate the structure of \textit{minimally motivating
    subgraphs}, the smallest subgraph where the agent completes the
task, for some fixed reward. Follow-up work of
\citet{tang2017computational}  shows that finding \textit{any}
motivating subgraph is NP-hard. Instead of deleting edges,
\citet{albers2019motivating} consider the problem of spreading a fixed
reward onto arbitrary vertices to motivate an agent to complete a
task, and find that this too is NP-hard (with a constrained budget).  
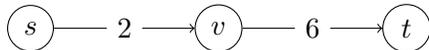
\begin{figure}[h]
    \centering
    \begin{tikzpicture}
	\begin{pgfonlayer}{nodelayer}
		\node [style=standard] (0) at (0.5, 0) {$s$};
		\node [style=label] (1) at (3, 0) {$2$};
		\node [style=standard] (2) at (10.5, 0) {$t$};
		\node [style=standard] (3) at (5.5, 0) {$v$};
		\node [style=label] (4) at (8, 0) {$6$};
	\end{pgfonlayer}
	\begin{pgfonlayer}{edgelayer}
		\draw (0) to (1);
		\draw [style=weighted directed, in=180, out=0] (1) to (3);
		\draw (3) to (4);
		\draw [style=weighted directed] (4) to (2);
	\end{pgfonlayer}
\end{tikzpicture}
    \caption{Let $(s, v)$ represent buying a gym membership and $(v, t)$ represent working out regularly for a month \citep{roughgarden2016cs269i}. At $t$, the agent receives a reward of $11$ due to health benefits. With bias $b = 2$, the agent initially believes this task is worth completing, but due to his bias, abandons the task at vertex $v$, after having already purchased the membership.}
    \label{fig:gym}
\end{figure}

These results focus on the problem of convincing an agent to complete a
task, rather than redirecting agents to cheaper paths. Though these
goals are related, it's natural to wonder how we might sway agents
towards more optimal behavior, rather than merely settling for task
completion. In other words, even if agents are willing to complete a
task using a needlessly expensive path (perhaps because of a large
reward), we should still consider how to make them behave more
optimally.  
\citet{kleinberg2016sophisticated} partially investigate this question in a model involving \textit{sophisticated} agents, who plan around their present bias. They consider several types of \textit{commitment devices} -- tools by which sophisticated agents can constrain their future selves. However, these tools may require more powerful agents or designers, and don't necessarily make sense for naive agents. 
\citet{saraf2020competition} takes a different approach, arguing that task designers can induce optimal behavior by setting up a competition between biased agents. While they obtain strong results for several families of graphs, there are also graphs where their competitive model can offer no benefit to agents.

Finally, \citet{kleinberg2014naive} consider a restricted version of our chunking problem, which is close to a special case of our model. They focus on the single edge graph $(s, t)$, and derive the optimal chunking in that setting. When considering general graphs, we obtain a similar result when chunking edges on the cheapest path; for other edges, the optimal chunking is more complex. Further, looking at general graphs allows us to ask how a fixed chunking budget should be best allocated across multiple edges, and, more broadly, how to convince agents to take a different (and cheaper) path. 

The rest of the paper is organized as follows.
In Section 2, we present a model for chunking and explain its
simplifying features. In Section 3, we focus on chunking single edges
within a graph. We first describe how chunking an edge $(u, v)$ can be
thought of as lowering the agent's present bias towards only that
edge. We then explore the structure of optimal edge chunkings, that
is, chunkings that
lower the agent's ``selective bias'' as much as possible. For edges
on the shortest path, we provide a closed form for the optimal
chunking. For other edges, optimal chunkings are considerably more
complex, but we provide an efficient algorithm to compute them. In
Section 4, we provide an algorithm to optimally distribute a fixed
number of chunks across multiple edges within a graph. 
In Section 5, we provide a tight bound on the cost
ratio for biased agents in terms of the number of chunks alloted to
the task designer. 
Our bound implies that with a linear number of chunks alloted to each edge, 
the cost ratio can be reduced to a constant factor. Finally, in Section 6, we consider the problem of
chunking a single task graph for two types of agents
simultaneously, where an agent's type is 
their bias. As an example, consider an instructor who wants a good
chunking for both rare and frequent procrastinators. We provide 
algorithms to chunk optimally under local and global budgets for two
types of agents. We also show how to extend our result to $m$ types of
agents, if we add the (simplifying) constraint that all agents must
take the same path through the graph.  

\section{Chunking Model}
We first explain the model of present bias in more detail. As
mentioned before, we start with a weighted, directed, acyclic graph
$G$ that represents a task, with start $s$ and end $t$. A
present-biased agent with bias parameter $b$ behaves as follows. Let
$c(v \to t)$ represent the cost of the shortest path from $v$ to $t$,
and let $c(u, v)$ represent the weight of edge $(u, v)$. From node
$u$, the agent goes to vertex $\argmin_{v: (u, v) \in E} bc(u,v) + c(v
\to t)$. We refer to $bc(u, v) + c(v \to t)$ as the agent's
\emph{perceived cost of starting with edge $(u, v)$ and then taking the shortest path to $t$}. 
We abbreviate this as \emph{the agent's perceived cost of starting
with $(u, v)$}. At each vertex, they go to the neighbor that
minimizes their perceived cost, continuing until they reach $t$.

We next consider chunking.
We distinguish two different settings where chunking helps:
\begin{enumerate}
    \item The task designer wants agents to take the cheapest path through the graph, rather than the more expensive path their bias would lead them to take.
    \item In a model where agents can abandon their path at any time
      (if the perceived cost is less than the reward),
      the task designer wants to prevent such abandonment. 
\end{enumerate}
We mainly focus on the first case in this paper, but our analysis easily extends to the abandonment setting. To investigate different models of chunking, consider the following graph, the $n$-fan (in which a biased agent can take an exponentially more expensive path than optimal \citep{kleinberg2014naive}):

\ctikzfig{n-fan1}
The task designer wants the agent to take the path $(s, t)$ instead of the longer path around the fan that an agent will take when their bias $b > c$. The simplest model of chunking allows them to break the edge $(s, t)$ into pieces as follows:

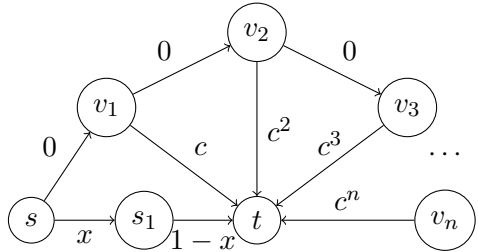
\ctikzfig{n-fan-split-1}
The designer gets to choose $x$ (i.e., they get to choose how much work
is done in the first and second chunk). Note that the intermediate
node $s_1$ doesn't have any connections, except to $t$. It's easy to
show that the best choice of $x$ is $0$, as this means that the
agent's present bias does not play any role in their decision (all
edges out of $s$ have 0 cost). From a different perspective, this
model seems to be taking advantage of the ``lock-in'' effect of
$s_1$ -- once the agent goes there, they cannot take an alternative
path, even though they did not actually do any work to get there. But
our intuition suggests that chunking a very difficult task into a cost
$0$ ``task'' followed by the same difficult task should not help
much. So, this doesn't seem to be a good model for chunking.  As an
aside, even if we require that $x$ is not too small, the obvious
solution for the task designer is to make the first chunk as small as
allowed -- there's not much interesting in this model.\footnote{If we
move to the abandonment setting, the task designer \textit{is}
incentivized to do a non-trivial split here; they would want to
balance the perceived costs of starting with edges in their chunking
in order to
avoid abandonment. However, the model we investigate induces a similar
balancing problem even without abandonment (and extends naturally
to the abandonment setting).} 

The more interesting model of chunking that we study breaks the edge $(s, t)$ into pieces as follows:

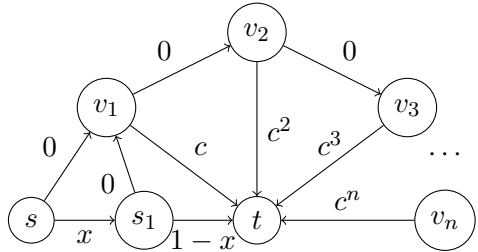
\ctikzfig{n-fan-split-2}
Here, the node $s_1$ keeps all the edges to other task nodes that $s$
had. This reflects the fact that even after completing a chunk, an
agent may decide to take another path to $t$ -- completing a chunk
doesn't ``lock'' an agent into a particular path. Of course, they will
be less likely to take another path if they finished a particularly
difficult chunk. Thus the model has the necessary tension -- the
designer wants to set $x$ high enough so that the agent actually still
takes the $(s, s_1, t)$ path, but not so high that they don't take
edge $(s, s_1)$ in the first place. Put another way, since the agent
can deviate at $s_1$, the designer wants to ensure that the perceived
costs of starting with  $(s, s_1)$ and with $(s_1, t)$ are both low. While we
have shown only 2-chunk examples, in our general model the task designer
splits an edge into $k$ chunks, 
whose costs sum to the original cost.\footnote{Note that in our
formalization, we remove the original edge for simplicity. However, if
we kept the original edge, the agent would never strictly prefer it,
no matter what the chunking. So it's mathematically equivalent to
think of the original edge still being there.  Moreover, this
interpretation maps better to our examples, where the task designer
does not actually enforce the chunking.} 

It is worth discussing three simplifying features of our model. First,
we assume that tasks can be arbitrarily split: each edge in the
chunking can have any cost, so long as the total cost remains fixed. A
more realistic model might constrain edges to have fixed chunking
options. For example, when chunking an essay, it could be the case
that each chunk must consist of some number of paragraphs; essays
cannot be chunked more finely. However, we believe that solving our
continuous relaxation will provide reasonable insight into the
discrete problem. Our informal argument is as follows: if the number
of potential chunks in the discrete problem is high, then our optimal
solution to the continuous version will be a good approximation. If
the number of potential chunks is low, then solving the discrete
problem is easy (there aren't many possible chunkings). Though we will
not consider the discrete version further, it would be interesting to
understand if
there are fundamentally different challenges in that setting. 

The second simplifying assumption is that the chunking ``overhead''
cost to the agent is zero. In other words, no matter how many chunks
an edge is split into, the total cost of that edge remains fixed
(notably, it does not increase).
In reality, there is
probably some cost to the agent per chunk. For instance, the agent
might stop working between chunks, and then have some cost associated
with getting back to work. We assume that this ``restarting'' cost is
very low relative to the other costs, and thus ignore it.
In any case, since each chunk gives the task
designer (weakly) more power in our model, we typically assume that
there is some given chunking budget $k$; if chunking instead had some
fixed overhead, there would exist an optimal $k$, as additional chunks
have diminishing returns but fixed overhead.

Lastly, we specify how agents break ties. If an agent at $u$ views
multiple neighbors as having the same perceived cost, the agent will
pick the neighbor that is part of a chunked path if exactly one
neighbor is part of a chunked path. Otherwise, they pick the first
vertex in some lexicographical ordering. This tie-breaking behavior is
mathematically convenient when constructing the optimal chunking, as
we can simply ensure that the perceived cost starting with each step in
the chunking matches the agent's otherwise best option. For a more
thorough treatment of tie-breaking rules in the base model of present
bias, see \citep{dementiev2021inconsistent}. 

We also contrast our model with a model of ``checkpoints''. As we mentioned, we consider chunking to be a purely mental tool to combat present bias. One might consider a stronger model, where the task designer (e.g., an instructor) can incentivize agents to complete a task in chunks. For example, the instructor might set an earlier (graded) deadline for the thesis statement of an essay. We can model this as the task designer having the power to split up the final reward $r$ onto intermediate vertices or edges, in addition to being able to chunk edges. Although we will not investigate this checkpoint model in this paper, we hope to investigate it in future work. While both the chunking model and checkpoint model are realistic choices to model classwork, we believe that the chunking model is a better fit for algorithmically chunking a user's to-do list; in that setting, the algorithm cannot enforce the chunks, but merely suggests them to the user. 

\section{Optimal Edge-Chunking}
In this section, we consider how to optimally chunk a single
edge. What do we mean by an \textit{optimal} chunking? As mentioned
earlier, we think of chunking as lowering an agent's selective bias
towards the chunked edge. In other words, for any chunking, an agent
with bias $b$ will take the chunked path from $u$ to $v$ if and only
if an agent with bias $b' < b$ towards edge $(u, v)$ (and bias $b$
otherwise) will take $(u, v)$ in the original graph. We say that such
a chunking \emph{induces a selective bias of $b'$ towards $(u,
v)$}.\footnote{When it is clear from context, we often leave the edge
unspecified.} So, by an \textit{optimal} chunking, we mean one in
which the agent's selective bias is brought as low as possible (given
a fixed bound $k$ on the total number of chunks).  

Our results show that as the number of chunks tends to infinity, the
selective bias tends to $1$ (i.e., unbiased behavior). Thus, the
number of chunks is a powerful parameter in our model; in the next
section we answer the broader question of how to best chunk is an
arbitrary task graph with a limited chunking budget. 


\subsection{Edges on the shortest path}
The problem of optimally chunking is subtly different for edges on
the shortest path (where ``shortest'' ignores bias) and edges on
other paths. We first consider the simpler case of edges on the
shortest path, and start with two chunks.

\begin{lemma} \label{lem:base_case}
    To optimally split an edge $(u, v)$ that is on the shortest path into
two chunks, the first chunk should be a $\frac{b-1}{2b-1}$ fraction of the
work. With this split, the agent will behave with a selective bias of
$\frac{b}{2-1/b}$.  
\end{lemma}
\begin{proof}
Suppose we chunk $(u, v)$ into $(u_1, u_2, v)$. First, note that, because $(u, v)$ is on the shortest path in the original graph, no matter how the edge is chunked, the optimal behavior from $u_2$ will be to go to $v$ -- this can only be cheaper than $(u, v)$ in the original graph. Thus, the perceived cost of starting with edge $(u_1, u_2)$ while at vertex $u_1$ is $bc(u_1, u_2) + c(u_2, v) + c(v \to t)$, as the agent naively believes they will behave optimally in the future. This is the only way that we use the fact that $(u, v)$ is on the shortest path. 

The designer wants to minimize the maximum of the perceived cost of
starting with $(u_1, u_2)$ and the perceived cost of starting with $(u_2,
v)$, to best ensure that the agent takes the chunked path. These
perceived costs are $bc(u_1, u_2) + c(u_2, v) + c(v \to t)$ and
$bc(u_2, v) + c(v \to t)$ respectively. 

Let $x = c(u,v)$ represent the total amount of work to be
chunked, and let $x_1$ and $x_2$ represent $c(u_1, u_2)$ and $c(u_2,
v)$ respectively. Note that $x_2 = x - x_1$. We now plug the $x$'s
into the expressions above to get perceived costs of
\begin{align*}
    bx_1 + x - x_1 + c(v \to t) \mbox{ and }\\
    b(x - x_1) + c(v \to t).
\end{align*}
We want to set $x_1$ to minimize the maximum of the two quantities.
That is, we choose $x$ so that
\begin{align*}
    &\argmin_{0 \le x_1 \le x} \max(bx_1 + x - x_1 + c(v \to t), b(x - x_1) +c(v \to t)) \\
    &= \argmin_{0 \le x_1 \le x} \max(bx_1 + x - x_1, b(x - x_1)) \\
    &= \argmin_{0 \le x_1 \le x} \max((b-1)x_1 + x, -bx_1 + bx). 
\end{align*}
Both expressions are linear functions of $x_1$, with the first
increasing and the second decreasing. The minimum of the maximum is
thus where they intersect, that is, when 
$$
  (b-1)x_1 + x = -bx_1 + bx.$$
Simple algebra then shows that 
$$  x_1 = \frac{b-1}{2b-1} x.$$
With this value of $x_1$, the perceived costs starting with $(u_1, u_2)$
and with $(u_2, v)$ are identical. The latter perceived cost is
\begin{align*}
    b(x-x_1) + c(v \to t) &=  bx \cdot \frac{b}{2b-1} + c(v \to t) \\
    &= \frac{b}{2-1/b} \cdot c(u, v)  + c(v \to t) \tag{since $x = c(u, v)$}.
\end{align*}
(It's easy to verify that the former perceived cost matches.) 
Thus, the agent with bias $b$ takes the path $(u_1, u_2, v)$ when an agent with bias $b^* = \frac{b}{2-1/b}$ would have taken $(u, v)$ in the original graph.
\end{proof}

We now state the following theorem, which extends the above results to $k$ chunks. We first state a more general version which will be helpful in the next section. The proof is in the appendix.
\begin{restatable}{theorem}{chunkOptEdge}
    \label{thm:chunk-opt-edge}
    Suppose we partition an edge $(u, v)$ of cost $x$ into $k$ chunks. Let $u_1, \dots, u_k$ represent the vertices in this chunking, and let $c(u_i, u_{i+1}) = x_i$, where, for $1 \le i \le k$, the $x_i$'s are defined below.
    \begin{align*}
            1\le i \le k: x_i = \frac{(b-1)^{k-i}b^{i-1}}{b^k - (b-1)^k} x.
    \end{align*}
    With this chunking, the agent has selective bias $\frac{1}{1 -
      \left(\frac{b-1}{b}\right)^{k}}$. If, with this chunking, the
    shortest path from $u_i$ to $t$ is through $u_{i+1}$ for all $i >
    1$, then this chunking is optimal. 
\end{restatable}

The following corollary immediately follows from this theorem.
\begin{corollary}
    For an edge $(u, v)$ on the shortest path, the chunking given
    in \autoref{thm:chunk-opt-edge} is optimal. 
\end{corollary}
\begin{proof}
  No matter how an edge on the shortest path is chunked, the
  shortest path from any chunk to $t$ must be through the next chunk,
  as chunking does not increase the total cost of the edge. This
  satisfies the condition in the theorem to get optimality. 
\end{proof}

The corollary says that the designer is not best served by evenly
splitting the cost between the edges -- the designer should lower the
cost of earlier edges. When they do so, the agent will behave as if
they had selective bias $\frac{1}{1 - \left(\frac{b-1}{b}\right)^{k}}$
in the original graph towards edge $(u, v)$ (while having bias $b$
towards all other edges). 

For a simple application of this corollary, suppose the agent's bias
is $2$. Then, splitting each edge on the shortest path once (so
$k=1$) causes the agent to behave as if they have bias $4/3$ on the
shortest path in the unmodified graph (and they still perceive other
edges with bias $2$). 

\subsection{Edges not on the shortest path; a motivating example}
We first motivate our results. For edges that are on the shortest
path, it's clear why a designer would want to chunk them -- they want
to convince agents to incur as little cost as possible. However, in
the next section we consider the natural problem where the
designer has a fixed chunking budget $k$. In such cases, our earlier
results imply that if the agent's bias is sufficiently high, it may
not be possible to convince them to stick to the shortest
path. However, the designer may be able to lower the agent's cost by
chunking other edges, which are not on the shortest path. Consider the
following graph as an example. 

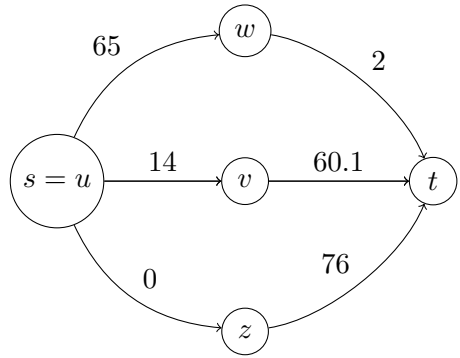
\ctikzfig{non-optimal-chunk-req}
Suppose that the agent has bias $2$. Let $P_w, P_v, P_z$ represent the
paths to $t$ through $w, v$, or $z$ respectively. The agent's bias causes
them to take $P_z$, the most expensive path. How should we best use a
fixed budget of 3 chunks to lower the agent's cost? First, note that
by \autoref{thm:chunk-opt-edge}, the optimal chunking of $(u, w)$
induces a selective bias of $8/7$. Even with this optimal chunking,
the agent would still prefer $P_z$, as $8/7 \cdot 65 + 2 > 76$. So, we
cannot lower the agent's cost by chunking $(u, w)$. Will chunking $(u,
v)$ instead help? 

Note that, for edges not on the shortest path (which we will
sometimes abbreviate to ``non short-path edges''), we could still apply the
chunking from \autoref{thm:chunk-opt-edge} to get the selective bias
described in that theorem. For $(u, v)$, \autoref{thm:chunk-opt-edge}
tells us to set $x_1 = 2, x_2 = 4$, and $x_3 = 8$, resulting in the
following graph. 

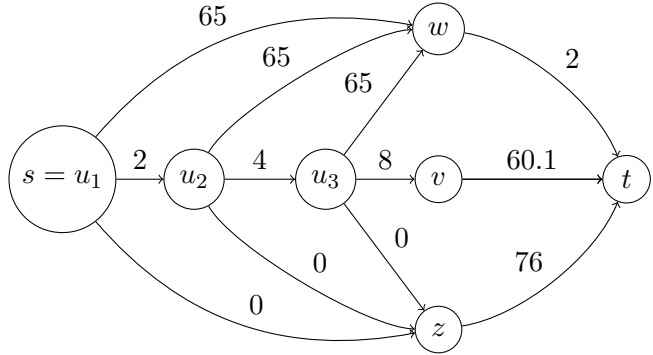
\ctikzfig{non-opt-chunked-new}
Under this chunking, the cheapest path from $u_1$, $u_2$, or $u_3$ to
$t$ all go through $w$. The agent's perceived costs of starting with
the edges in the chunking are, in order, $71, 75$, and $76.1$
(so the agent would take edge $(u_3, z)$ instead of sticking to the 
chunking). If $(u, v)$ was a shortest edge in the original graph (for
example, if $w$ did not exist), then the same chunking would have
identical perceived costs of $76.1$ starting with all edges. But when the
cheapest path from a chunked vertex to $t$ is through the external
vertex $w$, the perceived cost of starting with early edges
decreases. An optimal 
chunking should thus increase the cost of early edges and decrease the
cost of later edges to result in more balanced perceived costs. In the
example above, if we split the costs so that $c(u_1, u_2) = c(u_2,
u_3) = 3.55$ and 
$c(u_3, v) = 6.9$, then the cheapest path from $u_1$ or $u_2$ to $t$
is through $w$, while the cheapest path from $u_3$ to $t$ is through
$v$. Thus, the perceived costs of starting with the first edge and the
second edge are both $74.1$, and the perceived cost of starting with
$(u_3, v)$ is 
$73.9$. This is the optimal chunking of $(u, v)$, and it improves the
agent's cost by convincing them to take $P_v$ instead of $P_z$. Thus,
this example shows that we have good reason to chunk non short-path
edges, and our existing chunking results are insufficient for
such edges.  

\subsection{Optimally chunking for edges not on the shortest path}

As the example in the previous section suggests, it's important to
keep track of the shortest path from chunking vertices to $t$. Note
that if the shortest path from $u_i$ to $t$ is through $w$ rather than
$u_{i+1}$, then the shortest path from any $u_j$ to $t$, where $j <
i$, is also through $w$.

Thus, for any chunking, define $u_\tau$ as the \emph{transition vertex}:
the last vertex where the shortest path is through $w$, where $w$ is
the next vertex on the shortest path from $u$ to $t$ in the original
graph.
If the shortest path always follows the chunking, then define $\tau$
as $0$. On the other hand, if the shortest path is always through
external vertices, then $\tau = k$. For a shortest-path edge, all
chunkings have $\tau = 0$ (and thus the optimal chunking is given by
\autoref{thm:chunk-opt-edge}). But for non short-path edges, the optimal
chunking may have a higher value of $\tau$ (in the previous example,
the optimal chunking had transition vertex $\tau = 2$). Though the
case where $\tau = 0$ admits a nice closed form, in general we provide
an algorithm that determines the optimal chunking by trying all
possible values of $\tau$. 

We can think of $\tau$ as the smallest value such that,
for all neighbors $w$ of $u$, we have $c(u, w) + c(w \to t) \ge c(u_{\tau+1},
u_{\tau+2}) 
+ c(u_{\tau+2} 
\to t)$. We can rewrite this as follows, using the notation of
\autoref{thm:chunk-opt-edge}: 
\begin{equation}\label{eq1}
  \begin{array}{ll}
    c(u, w) + c(w \to t) &\ge c(u_{\tau+1}, u_{\tau+2}) + c(u_{\tau+2} \to t) \\
    &= x_{\tau+1} + \sum_{i=\tau+2}^k x_k + c(v \to t) \\
        &= x - \sum_{i=1}^{\tau} x_i + c(v \to t). 
  \end{array}
  \end{equation}
Let $\delta = x + c(v \to t) - (c(u, w) + c(w \to t))$ represent the
difference between the cost of the cheapest path from $u$ to $t$ through $v$
and the cost of the
cheapest path from $u$ through $w$ in the original graph (in the
previous example, $\delta = 74.1 - 67 = 7.1$). Then \autoref{eq1} is
equivalent to $\sum_{i=1}^{\tau} x_i \ge \delta$. 
For an edge on the shortest path, $\delta$ is negative, which
is why $\tau$ must be equal to $0$ for those edges.
Moreover, if $\delta \le x$, then it is possible to split the costs
among the edges to allow any choice of $\tau$: we simply put at least $\delta$
of the cost on the first $\tau$ edges 
while ensuring that the sum of costs of the first $\tau-1$ edges does not
exceed $\delta$. So, in addition to requiring that $\sum_{i=1}^{\tau}
x_i \ge \delta$, we also need $\sum_{i=1}^{\tau-1} x_i < \delta$.  

Before we get to our main result, we first introduce some more definitions and notation. 
Let $e_i = (u_i, u_{i+1})$ be the $i$th edge of a chunking, and let $p(e_i) = bx_i + c(u_{i+1} \to t)$ represent the perceived cost of starting with edge $e_i$. 
Let the \textit{bottleneck} of a chunking be the highest
perceived cost starting with any edge on that chunking (i.e. $\max_i
p(e_i)$). It's easy to see that the bottleneck of a chunking
determines the selective bias the chunking will induce; any agent who
will get past the bottleneck will complete the entire chunked path. So
an optimal chunking is a chunking with the smallest bottleneck. Finally, let a $k$-chunking of an edge be any chunking that splits the edge into $k$ chunks.

We now state some useful lemmas; their proofs can be found in the appendix.
\begin{restatable}{lemma}{chunkImprove} 
    \label{lm:chunking_improve}
    Suppose that $C$ is a chunking with bottleneck $\beta$. If another
    chunking $O$ has bottleneck $\beta' < \beta$ and the same
    transition vertex $\tau$, then $O$ must lower the cost of all
    edges that are bottlenecks in $C$, and thus raise the cost of the
    remaining edges. 
    
\end{restatable}
Though the lemma seems obvious at first glance, it relies crucially on
the fact that $C$ and $O$ have the same transition vertex $\tau$. It's
possible for $O$ to not lower the cost of all edges that are
bottlenecks in $C$ but still get a lower bottleneck cost if $O$ has a
different transition point. But with $\tau$ fixed, the difference
between the perceived costs starting with any edge in $C$ compared to
$O$ depends only on the cost the chunkings assign to the edge. 
\begin{restatable}{lemma}{chunkingOpt}
    \label{lm:chunking_opt}
    If a chunking $C$ has the same perceived cost starting with
    any edge in the chunking, then $C$ is optimal. 
\end{restatable}
\autoref{lm:chunking_opt} guides the algorithm, which tries to ensure
that the perceived costs starting with edges in $C$ are as close as
possible. At a high level, the algorithm enumerates over all values of
$\tau \in \{1, \dots, k\}$. We start with a chunking where the first
$\tau$ edges are assigned cost $\delta/\tau$, which ensures that they
all have the same perceived cost $\alpha$. We then use
\autoref{thm:chunk-opt-edge} to distribute the remaining cost over the
last $k-\tau$ edges, which also equalizes their perceived cost to some
$\beta$. If $\alpha \ge \beta$, we argue that this chunking is optimal
for the fixed $\tau$. Otherwise, we make some local updates to the
chunking, which brings $\beta$ as close to $\alpha$ as possible while
maintaining the invariant that $\beta \ge \alpha$. The full description of this algorithm, \autoref{alg:opt-edge-chunk}, can be found in the appendix.

\begin{restatable}{theorem}{algChunkingOpt}
    \label{thm:algChunkingOpt}
    Given any edge $(u, v)$, we can determine the optimal $k$-chunking
    in $O(k)$ time, assuming that the shortest paths from $u \to t$
    and $v \to t$ have been precomputed. 
\end{restatable}
\begin{proof}[Proof Sketch]
    For a fixed $\tau$, we start by setting $x_1 = x_2 = \dots = x_\tau = \delta/\tau$, and chunk the remaining $x - \delta$ cost over the remaining $k - \tau$ edges according to \autoref{thm:chunk-opt-edge}. Doing so ensures that $p(e_i) = \alpha$ for all $i \le \tau$ and that $p(e_i) = \beta$ for all $i > \tau$ ($\alpha$ and $\beta$ are defined in the appendix). If $\alpha = \beta$, by \autoref{lm:chunking_opt} we're done. In the case where $\alpha > \beta$, we show that we're done for this fixed $\tau$.

    The case where $\beta > \alpha$ is the bulk of the proof. The key is that $p(e_\tau)$ can be grouped into \textit{either} the earlier or later edges. Since $\beta > \alpha$, we carefully increase the cost of the first $\tau - 1$ edges and decrease the cost of the later edges to produce the optimal chunking for this value of $\tau$.
\end{proof}
\DontPrintSemicolon

\section{Optimal Chunking in Task Graphs}
In the previous section, we focused on optimally chunking a single
edge. One reason \textit{why} a task designer might want to do that is
to convince agents to take much cheaper paths through the graph, by
chunking the right edges. In this section, we assume that the designer
can chunk any edge in the graph, but can place only a limited number
of chunks (their chunking ``budget''). Which edges should they chunk
to ensure that the present-biased agent takes as cheap a path as
possible, and how should they chunk those edges? 

We first answer the latter question. Is lowering the agent's selective
bias towards an edge as much as possible (i.e., optimally chunking
that edge) always the best way to reduce their overall cost? Though
this might seem obviously true, a surprising fact is that a
present-biased agent's cost is not monotone in their bias; a smaller
bias may sometimes increase their total cost
\citep{kleinberg2016sophisticated}. Despite this, when trying to
minimize the agent's cost, the designer should optimally chunk any
edge they want to chunk (e.g., by using
\autoref{alg:opt-edge-chunk}). The only challenge is in finding 
which edges to chunk. 

To see why this is true, first note that chunking an edge $(u, v)$
will not change its overall cost, and thus will not impact the agent's
decisions unless they are at $u$. Second, it's easy to see that
chunking cannot \textit{increase} one's selective bias, as no edge in
the chunking can have more cost than the original edge cost. Thus, any
chunking of edge $(u, v)$ serves to convince the agent to take
$(u, v)$. And the best way to accomplish that is to minimize the
agent's perceived cost starting with that chunked edge, which is
exactly what an optimal edge-chunking does.  

\subsection{Local Constraints}
We consider two types of constraints on the designer. We call the
first a \emph{local} constraint; in this case the designer can break
any set of edges into up to $k$ chunks, for some parameter $k$.  If we
think of edges as representing relatively large subtasks, then this
just says that any relatively large subtask can be split into up to
$k$ smaller subtasks.  We call the second a \emph{global} constraint:
in this case, the designer gets a budget of $k$ chunks, and can use no
more than $k$ chunks altogether.

In this section we consider local constraints.
A naive approach would be to just optimally chunk every edge into $k$
chunks, using our earlier results. But this wouldn't necessarily give the
best overall chunking for the graph. Why not? The intuition is that we
want the agent's perceived cost of the path that the designer actually
wants the agent to use to be low. We are better served by \textit{not} 
chunking edges away from this path, so that the agent is not tempted
to deviate. So at a high level, the algorithm first figures out the
cheapest \textit{feasible} path for the agent (given $k$), and then
uses the optimal edge-chunking algorithm to actually chunk this path. 

\begin{theorem}\label{thm:local-budget}
    Given any task graph $G = (V, E)$ and a local constraint $k$,
    we can optimally chunk $G$ with at most $|E|$ applications of
    \autoref{alg:opt-edge-chunk}, for a total runtime of $O(|E|k +
    |V|)$. 
\end{theorem}
\begin{proof}
    First, we can use well-known algorithms to find the costs of the shortest path from any node to $t$ in time $O(|E| + |V|)$, since $G$ is a directed, acyclic graph \cite{cormenAlgorithms}.
Given a vertex $u$, let $w = \argmin_{v: (u, v) \in E} bc(u, v) +
    c(v \to t)$ be the vertex that the present-biased agent would go
    to without any chunking. Further, let $\alpha_u = p(u, w)$ be the
    perceived cost of starting with edge $(u, w)$. Let $v \neq w$ be
    an arbitrary out-neighbor of $u$ (i.e., a vertex $v$ such that there
    is an edge $(u,v)$). \autoref{alg:opt-edge-chunk}
        gives us the lowest possible bottleneck cost of a $k$-chunking
    of $(u,v)$; denote this as $\beta_{u, v}$. If $\beta_{u, v} \le
    \alpha_u$, the agent can be made to take $(u, v)$. If not, then
    they won't take $(u, v)$ under any $k$-chunking. 

The algorithm is straightforward. At every vertex $u$,
    determine $\alpha_u$ as well as $\beta_{u, v}$ for all
        out-neighbors $v$ of $u$. If $\beta_{u, v} > \alpha_u$, remove edge $(u,
    v)$ from the graph. Call the resulting graph $G'$. Then, simply
        compute the shortest path in $G'$, and chunk every edge on that
    path with \autoref{alg:opt-edge-chunk}. 

    There will always be an $s$-$t$ path in $G'$, as the edges the
    agents would take without chunking can never be removed. By
    construction, the path in $G'$ that we chunk is one that the agent
    will take in $G$ after chunking. Finally, there can be no cheaper
        path, as we remove only edges that the agent cannot be convinced
    to take. 
\end{proof}
We briefly discuss a different perspective on the algorithm above,
which will be useful when comparing to the results of the next section.
We can think of the algorithm as a
dynamic program with the following recurrence: 
\begin{align*}
    cost[u] = \min_{\mathclap{v: (u, v) \in E, \beta_{u, v} \le \alpha_u}}
        c(u,v) + cost[v]. 
\end{align*}
Here, $cost[u]$ is the cost of the cheapest $u$ to $t$ path we can convince the agent to take, and the base case is simply $cost[t] = 0$. This recurrence is exactly the recurrence that a shortest-path algorithm solves, except for the condition that $\beta_{u, v} \le \alpha_u$. Thus, the first part of the algorithm simply removes edges that do not satisfy this condition, and then the solution to the shortest path problem will solve the above recurrence.

\subsection{Global Chunking Budget}
In this section we consider global constraints; the designer
must consider where to best allocate chunks to have the most
impact. As before, we can use the optimal edge-chunking algorithm to
solve this problem; only marginally more computation is required. 

\begin{theorem}\label{thm:global-budget}
    Given any task graph $G = (V, E)$ and a global constraint
    $k$, we can determine the optimal chunking configuration with at
    most $O(|E|\log k)$ applications of \autoref{alg:opt-edge-chunk},
    for a total runtime of $O(|E|k \log k + |V|)$. 
\end{theorem}
\begin{proof}
    As before, we first compute the cost of the shortest path from any node to $t$ in time $O(|V| + |E|)$.
        For a local budget, we sorted edges into feasible and infeasible
    edges, where an edge was feasible if we could convince the agent
    to take it with at most $k$ chunks. Here, we instead determine the
    minimum number of chunks that's necessary for an agent to take
    each edge (if the number is at most $k$). Since the optimal
    bottleneck cost is decreasing in the number of chunks $k$, we can
    simply use binary search to find this minimum number. 

    In more detail, let $u$ be an arbitrary vertex and define
    $\alpha_u$ as above. For any out-neighbor $v$ of $u$, let
    $\beta_{u, v}^l$ be the lowest possible bottleneck cost of any
    $l$-chunking of $(u,v)$. Let $l_{u, v}$ be the smallest $l \le k$
    such that $\beta_{u, v}^l \le \alpha_u$. If no such $l$ exists,
    then $l_{u, v} = \infty$. $l_{u, v}$ can be computed in $O(\log k)$
    applications of \autoref{alg:opt-edge-chunk} with binary search,
    since $\beta_{u, v}^l$ is decreasing in $l$. 

   Now let $cost[u, i]$ denote the cost of the cheapest path from 
    $u$ to $t$ that we can convince the agent to take with at most $i$
    chunks. The base case is simply $cost[t, i] = 0$ for all $0 \le i
    \le k$. The recurrence is as follows. 
    \begin{align*}
        cost[u, i] = \min_{\mathclap{v: (u, v) \in E, l_{u, v} \le i}}
                c(u, v) + cost[v, i - l_{u, v}]. 
    \end{align*}
    The final solution is $cost[s, k]$. The correctness of this recurrence follows from the fact that $l_{u, v}$ is the smallest number of chunks needed to convince the agent to take edge $(u, v)$. For the runtime, note that it takes $O(E k \log k)$ to compute $l_{u, v}$ for all $(u, v) \in E$. For the recurrence, the $\min$ considers $|E|$ possibilities for each value of $i \in \{0, \dots, k\}$, for a total runtime of $O(|E|k)$. Finally, to actually compute the recurrence, we can simply proceed backwards through some topological ordering of the graph. 
\end{proof}

\section{Optimizing the cost ratio}
Define the cost ratio of a present-biased agent to be $C_b(s \to
t)/c(s \to t)$, where $C_b(s \to t)$ is the cost that a present-biased
agent with bias parameter $b$ incurs in the graph, and $c(s \to t)$ is
the shortest path cost. The goal of this section is to understand how
the cost ratio of the present-biased agent decreases as the task
designer places more chunks in the graph. Put another way, in the
previous section we provided algorithms that optimally chunked task
graphs, given a fixed chunking budget $k$. Here, we prove performance
guarantees on those algorithms, where the algorithm's ``performance''
is measured in how much it reduces the cost of the agent's path.  

Existing results have characterized the worst-case cost ratio over all task graphs.
\begin{theorem}[Adapted from \citet{tang2017computational}]
    \label{thm:tang}
    The cost ratio for an agent with present bias $b$ is at most $b^n$, over all task graphs. The $n$-fan (see \autoref{fig:n-fan}) can get arbitrarily close to this cost ratio as $c$ approaches $b$ from below.
\end{theorem}
\newcommand{\bmin}{\ensuremath{b_{\text{min}}}}
We want to characterize the worst-case cost ratio after chunking. More precisely, we consider the following question. Let $G$ be arbitrary, and let $G'$ denote an optimal $k$-chunking of $G$. What is the worst-case cost ratio for $G'$? 
We start by considering local constraints; thus, $G'$ is the result of
breaking an arbitrary number of 
edges in $G$ into at most $k$ chunks.  
Let $\bmin$ be the
selective bias guaranteed by \autoref{thm:chunk-opt-edge}. That is,
let: 
$$\bmin = \frac{1}{1-\left(\frac{b-1}{b}\right)^k}.$$

\begin{theorem}
    \label{thm:costRatio}
If $G'$ is an optimal chunking of $G$ with local constraint $k$,
then the cost ratio 
    for an agent with present bias $b$ in $G'$ is at most $\bmin^n$. 
\end{theorem}
\begin{proof}
    We simply chunk every edge into $k$ chunks using the chunking
    given in \autoref{thm:chunk-opt-edge}, which results in the agent
    viewing every edge with a selective bias of $\bmin$. Call the
    resulting graph $G''$. By the definition of selective bias, for
    every edge $(u, v) \in G$, an agent with bias $\bmin$ would go from
    $u$ to $v$ if and only if the agent with bias $b$ would
    traverse the chunking $(u_1, u_2, \dots, u_k, v)$ in $G$. Since
    this holds for every edge, the agent will incur exactly the same
    cost as an agent with bias $\bmin$ would incur in $G$. So by
    \autoref{thm:tang}, they incur cost at most $\bmin^n$ in $G''$,
    with bias $b$.  

    The theorem follows from the fact that $G'$ is an optimal chunking of $G$, so the agent will only do better there as compared to $G''$.
\end{proof}
\begin{restatable}{corollary}{costRatioCor}
Given a local constraint $k
    = O(n)$, the optimal chunking $G'$ of $G$ 
has constant cost ratio. 
\end{restatable}
\begin{proof}
    The proof involves only arithmetic after applying \autoref{thm:costRatio}. Details can be found in the appendix.
\end{proof}
The corollary shows that we can get an exponential reduction in the
agent's worst-case cost with only a linear number of chunks on every
edge, demonstrating the power of chunking. However, from a different
perspective, the bound in \autoref{thm:costRatio} seems weak. We
showed earlier that it's never necessary to chunk two edges leading out of the
same vertex, but here we chunk all edges. Further, we chunk every edge
with \autoref{thm:chunk-opt-edge}, despite that chunking not being
optimal for non short-path edges. Despite these concerns, the bound in
the theorem is tight, as demonstrated by chunking the $n$-fan. 

\begin{figure}[h]
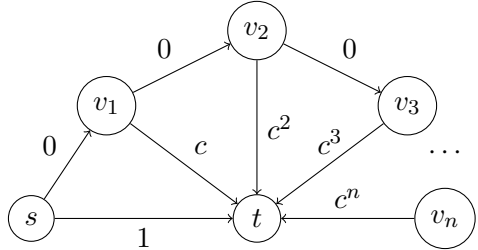

    \ctikzfig{n-fan1}
    \caption{This graph is the $n$-fan. If $c < b$, the agent will
      prefer edge $(v_i, v_{i+1})$ to $(v_i, t)$ for all $i$. Thus,
      the agent goes all the way around the fan, and incurs cost
      $c^n$.} 
    \label{fig:n-fan}
\end{figure}

\begin{lemma}
If $G$ is an $n$-fan with $c < \bmin$ and $G'$ is an optimal
    chunking of $G$ given local constraint $k$, then the
    cost ratio for an agent with present bias $b$ is $c^n$ in $G'$. 
\end{lemma}
\begin{proof}
    Let $G'$ be constructed by chunking every edge in the $n$-fan via
    \autoref{thm:chunk-opt-edge} (we can ignore the $0$ cost edges, as
    chunking a $0$ cost edge has no impact on the agent's
        decisions). In $G'$, the agent acts as if they had bias $\bmin$ in
    $G$. And such an agent would incur cost $c^n$ by going all the way
        around the fan, since $c < \bmin$. It remains to show 
        that $G'$ is an optimal chunking of $G$. 

    In fact, we show the stronger claim that any chunking of $G$ with a local budget of $k$ is (weakly) optimal, as no such chunking can cause the agent to take a cheaper path. To see this, suppose there is a chunking $G^*$ of $G$ such the agent goes from $v_i$ to $t$, for $i < n$ (this is the only way they could take a cheaper path). Then, $G^*$'s chunking of edge $(v_i, t)$ must have lower bottleneck cost than in $G'$. We claim that this is impossible, because \autoref{thm:chunk-opt-edge} will give the optimal chunking for edge $(v_i, t)$. To see this, notice that the shortest path from $v_i$ to $t$ is through edge $(v_i, t)$, which is exactly when \autoref{thm:chunk-opt-edge}'s chunking is optimal. As a result, no $G^*$ exists, and so $G'$ is an optimal chunking. 
\end{proof}

We have provided a tight characterization for the worst-case cost
ratio in terms of the number of chunks given a local constraint.
We conjecture that a similar result extends to global constraints.
Let $k$ be the global chunking
budget. Clearly, we could get an upper bound on the worst-case cost
ratio similar to that of \autoref{thm:costRatio} by evenly splitting
the chunks so that each edge satisfies a local constraint of $k/m$,
where $m = |E|$. We conjecture that this would also be an
asymptotically tight bound, as it seems that the optimal chunking in
the $n$-fan would need to spread chunks evenly among half the edges
(i.e., the edges $(v_i, t)$). 

\section{Optimal Chunking for Multiple Agents}
We now consider the problem of chunking a task graph for two types
of agent, where an agent's type is their bias. For example, an
instructor might reasonably expect some students to procrastinate
rarely and others to procrastinate frequently. Yet the instructor cannot
chunk the task separately for different students (indeed, they may
well not know a given student's type).  How should
they chunk the task while balancing the cost that both types of
students incur? We answer this question in two settings. We first show
how to optimally chunk the graph for two types of agents, $A_1$ and
$A_2$ with $b_1 < b_2$. Second, we show how to optimally chunk the
graph for $m$ types of agents, with the additional constraint that all
agents take the same path. Allowing agents to take different paths
gives the designer more power but also makes the problem significantly
more complex to analyze; removing this possibility allows us to design
for $m$ types, rather than 2. 

Note that in the case of a single agent, there is an obvious way to define the ``optimal'' way to chunk an edge -- it's the one that agent perceives as cheapest. This definition is also useful for chunking the task graph optimally, as it tells us which edges we can persuade the agent to take. With two agents, it's unclear what it would mean to ``optimally'' chunk an edge. An intuitive definition would be that the optimal chunking for an edge minimizes the average perceived cost of the two agents. But that is wholly unhelpful for graph chunking, as it doesn't tell us which edges we can persuade either agent to take. So, we instead consider two related problems: convincing agents to take the same path, and convincing agents to split up. Solving these two problems will allow us to chunk the task graph while minimizing the sum of the agents' costs. 

\subsection{Splitting Agents onto Separate Paths}
In this section, we want to find the chunking $C^*$ of $(u, v)$ such that $A_1$ takes $C^*$ and $A_2$ finds $C^*$ ``maximally unappealing'': formally, $C^*$ has the maximum perceived cost for $A_2$ over all chunkings $A_1$ would take. We can use such a chunking to split up two agents who are both at the same vertex. We start by defining some terms. Let $p(e; b_i)$ represent the perceived cost of edge $e$ for the agent with bias $b_i$.  Here, agent $A_1$ has bias $b_1$, and agent $A_2$ has bias $b_2$, where $b_1 < b_2$. Then, let $\alpha_u^{(i)}$ represent $A_i$'s perceived cost of their best option at $u$ (without chunking). So, $\alpha_u^{(i)} = p(u, w_i; b_i)$, where $w_i = \argmin_{v:(u, v) \in E} b_ic(u, w_i) + c(w_i \to t)$.

We now describe the algorithm that solves this problem, \autoref{alg:chunkSplit}, at a high level; a full description can be found in the appendix. \autoref{alg:chunkSplit} first computes $C^*_1$, the optimal chunking of $(u, v)$ for $A_1$.\footnote{The algorithm does not rely on starting with an \textit{optimal} chunking; any chunking that $A_1$ takes would work.} Then, the algorithm iterates over all choices of $e_i$ and raises $p(e_i; b_2)$ as much as possible while ensuring that $A_1$ still takes the chunking. It does so by ``siphoning'' cost from other edges in the chunking onto $e_i$. It repeats this process for all choices of $e_i$.  This siphoning has three phases. 

In the first phase, we siphon from $x_{i-1}, \dots, x_1$ to $x_i$.\footnote{To ease exposition, we can think of ``siphoning'' as a continuous process where one cost is decreased as another increases. In practice, how much to siphon can be computed in $O(1)$ time; see the appendix for details.}  In the second phase, we siphon from $x_{i+1}, \dots, x_k$ to $x_i$. These phases are very straightforward, and terminate when $p(e_i; b_1) = \alpha_u^{(1)}$, where $\alpha_u^{(1)}$ is the perceived cost of the best alternative to $(u, v)$ from $A_1$'s perspective. In the third phase, we decrease $x_{>i}$ and increase $x_{\le i}$; because $b_1 < b_2$, doing this results in increasing $p(e_i; b_2)$ without increasing $p(e_i; b_1)$. 

Call the resulting chunking $C_i$. Note that $A_1$ will surely take $C_i$: $A_1$ took the original chunking, and all edges which were increased (potentially all $e_{\le i}$) were not increased beyond $\alpha_u^{(1)}$. We first prove the following conditions of the algorithm.
\begin{lemma}\label{lm:postconditions}
    Let $C_i = (e_1, \dots, e_k)$ be the chunking produced by
    iteration $i$ of \autoref{alg:chunkSplit}. Then: 
    \begin{enumerate}[(a)]
        \item $\sum_{j \neq i} x_j > 0 \implies p(e_i; b_1) = \alpha_u^{(1)}$
        \item $\sum_{j > i} x_j > 0 \implies \forall j \le i, p(e_j; b_1) = \alpha_u^{(1)}$
    \end{enumerate}
\end{lemma}
\begin{proof}
    For (a), if any $x_j >0$, then the algorithm terminated early in phase 1 or phase 2, which implies that $p(e_i; b_1) = \alpha_u^{(1)}$. For (b), if more could be siphoned from $x_{>i}$, then the algorithm would siphon more in phase 3, unless no edges in $e_{<i}$ can be increased further.
\end{proof}
The following theorem says that $A_2$ finds edge $e_i$ in $C_i$ maximally unappealing over all chunkings $A_1$ would take; the proof is in the appendix.
\begin{restatable}{theorem}{chunkSplitThm}
    \label{thm:chunk_split}
If $C_i$ is the output of the $i$th iteration of
    \autoref{alg:chunkSplit} and $C'$ is another chunking such that
    $p(e'_i; b_2) > p(e_i; b_2)$, then $A_1$ will not take $C'$. 
\end{restatable}

The theorem can be applied to show that our algorithm is correct. Let $C^*$ be the chunking with the maximum perceived cost from $A_2$'s perspective that $A_1$ will still take. Let $i^*$ be the bottleneck of $C^*$ for $A_2$. Then, the contrapositive of the theorem shows that our algorithm will find $C^*$ (or a chunking with equivalent $A_2$-perceived cost) when $i = i^*$.

Unfortunately, this problem is not symmetric with respect to $A_1$ and $A_2$. In other words, we still must solve the problem of chunking an edge such that $A_2$ takes it but $A_1$ finds it maximally unappealing. The only modification we need to make is to phase 3, where we instead \textit{increase} $x_{>i}$ and decrease $x_{\le i}$, which will increase $p(e_i; b_1)$ without increasing $p(e_i; b_2)$. More details can be found in the appendix.

\subsection{Keeping Agents on the Same Path}
In this section, we consider the problem of chunking a single edge $(u, v)$ so that all agents take the chunking. This problem can solved greedily, even if we have $m$ types of agents.
\DontPrintSemicolon
\setcounter{algocf}{2}
\begin{algorithm}
    \For{$i = k$ \KwTo $1$} {
        maximize $x_i$ such that $p(e_i; b_j) \le \alpha_u^{(j)}$ for all $j \in [m]$\\
        \If{$x_i < 0$}{
            \Return $\bot$
        }
        \If{$\sum_{i} x_i \ge x$} {
            lower $x_i$ so that $\sum_{i} x_i = x$ \\
            \Return{chunking $C$}
        }
    }
    \Return $\bot$ \tcp*[f]{$\sum_{i} x_i < x$}
    \caption{Greedily chunk edge $(u, v)$ into $k$ chunks for $m$ agents}
    \label{alg:greedy2}
\end{algorithm}
This algorithm will produce a chunking that the agents will all take, iff such a chunking exists. We use the following lemma, which is proven in the appendix. To introduce the lemma, we define a \textit{partial chunking} as a chunking that does not assign all the cost of the original edge. \autoref{alg:greedy2} can be viewed as building partial chunkings into a complete chunking.
\begin{restatable}{lemma}{mAgentsGreedy}
\label{lm:m_agents_greedy}
    Let $C$ and $C'$ be two (possibly partial) chunkings of the same edge. Suppose that $\sum_{i = l}^k x_i' > \sum_{i = l}^k x_i$. Then, there exists an $i \in [l, k]$ such that for all $b > 1$, $p(e_i'; b) > p(e_i; b)$.
\end{restatable}
The lemma says that if a chunking $C'$ assigns more cost to the last $k-l$ edges than $C'$, then one of those last $k-l$ edges must have a higher perceived cost (for \textit{any} present-biased agent). We now prove that the algorithm is correct.

\begin{theorem}
    \label{thm:samePath} \autoref{alg:greedy2} runs in time $O(mk)$. Further:
    \begin{enumerate}[(a)]
        \item If \autoref{alg:greedy2} returns a chunking $C$, then all agents will take $C$.
        \item If \autoref{alg:greedy2} returns $\bot$, then no chunking exists that all agents would take.
    \end{enumerate}
\end{theorem}
\begin{proof}
    Statement (a) is obvious; if a chunking is returned, then it must be the case that $p(e_i; b_1) \le \alpha_u^{(j)}$ for all $i$ and for all $j$. Thus, every chunk is more appealing than every agent's best outside option, and so all agents take $C$. The runtime is also obvious: inside the loop, the only work being done is computing the maximum $x_i$ such that $p(e_i; b_j) \le \alpha_u^{(j})$, for all $j \in [m]$.

    We prove statement (b) by looking at two cases. For the first
    case, suppose the algorithm returns $\bot$ at iteration $i$. This
    means that when $x_i = 0$, $p(e_i; b_j) > \alpha_u^{(j)}$ for some
    agent $j$. However, note that if $x_i = 0$, then $p(e_i; b_j) =
    \min(c(u, w) + c(w \to t), \sum_{l > i} x_l + c(v \to t))$, where
    $\sum_{l > i} x_l \le x$ (or the algorithm would have terminated
        at $i+1$). Over all chunkings, the smallest perceived cost of  
    the first edge is achieved when no weight is placed on it. Let
    $e_1^{\min}$ be the first edge in such a chunking. Then,
    $p(e_1^{\min}; b_j) = \min(c(u, w) + c(w \to t), x + c(v \to
    t))$. Since $x \ge \sum_{j > i} x_j$, we know that $p(e_1^{\min};
    b_j) \ge p(e_i; b_j) > \alpha_u^{(j)}$. Thus, in any other
    chunking, the agent $j$ would deviate at the first chunk. 

    In the second case, suppose the algorithm returns $\bot$ at the end. This means that, for all $i$, $p(e_i; b_j) = \alpha_u^{(j)}$ for some agent $j$ and $\sum_{i} x_i < x$. In other words, the chunking $C$ that the algorithm produces is a partial chunking, and a complete chunking must assign more cost. However, \autoref{lm:m_agents_greedy} says that if any chunking $C'$ assigns more cost, then there would be some edge $e'$ of $C'$ which all agents would perceive as more expensive. So, some agent would abandon their path at $e'$. Thus, there is no complete chunking that all agents would take.
\end{proof}

\subsection{Optimal Graph Chunking for Multiple Agents}
We now revisit the problem of optimal graph chunking, with a local or global chunking budget, $k$.
\subsubsection{Two Types}
We first assume we have a local chunking budget of $k$ chunks per edge, and try to minimize the sum of the two agents' (real) costs.\footnote{It's trivial to modify the recurrence to instead minimize the maximum of the two types' costs, a weighted average (useful if one type is much more common), or many other such functions.} We first reformulate our solution to the single agent case to introduce the idea of ``persuadable'' edges. In that case, we used the recurrence $cost[u]$ to represent the minimum cost of any $u \to t$ path that we could persuade the agent to take. We computed the recurrence via $cost[u] = \min_{v \in \mathcal{P}(u)} c(u, v) + cost[v]$, where $\mathcal{P}(u) = \{v: (u, v) \in E, \beta_{u, v} \le \alpha_u\}$ represents the set of vertices we can persuade the agent to take from $u$.

We can define a very similar recurrence for two agents. Say that two
paths $P$ and $Q$ are $(A_1, A_2)$-compatible if we can chunk (some
of) the edges along $P$ and $Q$ such that $A_1$ takes $P$ and $A_2$
takes $Q$. Let $cost[u, y]$ represent the minimum sum of the costs of
any $u \to t$ path $P$ and a $y \to t$ path $Q$ such that $(P, Q)$ are
$(A_1, A_2)$-compatible. Further, let $\mathcal{P}(u, y)$ be the set
of all edges $(v, z)$ such that $(u, v)$ and $(y, z)$ can be
``compatibly-chunked''. This means that, if $(u, v) = (y, z)$, then
there exists a chunking of $(u, v)$ that both agents take. Otherwise,
there exist chunkings $C_1, C_2$ of $(u, v)$ and $(y, z)$ such that
$A_1$ takes $C_1$ and $A_2$ takes $C_2$. If $u \neq y$ (i.e., the
agents start at different vertices), then $\mathcal{P}(u, y)$ can be
easily computed via the algorithms in Section 4. And $\mathcal{P}(u,
u)$ can be computed via the algorithms in Section 6.1 and 6.2.  

With these functions, the recurrence can be broken into three cases. The first case is when $A_2$ is about to go to the vertex, $u$, that $A_1$ is currently at. In this case, we need to ensure that our chunking of $(u, v)$ for $A_1$ doesn't cause issues for $A_2$. This case can be represented as:
\begin{align*}
    C_1(u, v, y) = \begin{cases*}
        c(y, u) + cost[u, u] & if $(v, u) \in \mathcal{P}(u, y)$ \\
                \infty & otherwise.
    \end{cases*}
\end{align*} 
The second case is similar, but with the agents flipped.
\begin{align*}
    C_2(u, y, z) = \begin{cases*}
        c(u, y) + cost[y, y] & if $(y, z) \in \mathcal{P}(u, y)$ \\
                \infty & otherwise.
    \end{cases*}
\end{align*}
Finally, if neither of the previous cases occur, the cost is:
\begin{align*}
    C_3(u, v, y, z) = c(u, v) + c(y, z) + cost[v, z].
\end{align*}
Putting it all together, the recurrence is:
\begin{align*}
    cost[u, y] = &\min_{(v, z) \in \mathcal{P}(u, y)} \min(C_1(u, v,
        y), C_2(u, y, z), C_3(u, v, y, z)).
\end{align*}
We first prove the correctness of this recurrence.
\begin{lemma}
    \label{lm:costRec}
    The recurrence for $cost[u, y]$ above is the cost of the cheapest paths $P: u \to t$ and $Q: y \to t$ such that $P$ and $Q$ are $(A_1, A_2)$-compatible.
\end{lemma}
\begin{proof}
    Assume that $cost[v, z]$ have been correctly computed for all $v$
    (resp. $z$) that are out-neighbors of $u$ (resp. $y$). We know
    that $u \neq v$ and $y \neq z$, because there are no self-loops in
        a DAG. We now proceed by cases. 
    
    \paragraph{Case 1: $u = y$.} First, note that $v \neq u, z \neq u$ for all $(v, z) \in \mathcal{P}(u, u)$. So, we will only be in the first case of the min.
    In this case, $P(u, y) = P(u, u)$ will return all $(v, z)$ such that there exist chunkings $C_1$ of $(u, v)$ and $C_2$ of $(u, z)$ such that $A_1$ takes $C_1$ and $A_2$ takes $C_2$, if both are at $u$. Further, if $v = z$, then $C_1 = C_2$ (i.e., $(u, v)$ is chunked such that both agents take it). Recall that $cost[v, z]$ is the cheapest cost of $(A_1, A_2)$ compatible paths $P': v \to t$ and $Q': z \to t$. Since $A_1$ going from $u \to v$ is compatible with $A_2$ going from $u \to z$, we get that the paths $P: (u, v) \cup P'$ and $Q: (u, z) \cup Q'$ are $(A_1, A_2)$ compatible.
    \paragraph{Case 2: $u \neq y$.} When $u \neq y$, all three cases of the min are possible. Since $P(u, y)$ describes all possible ways to chunk for $A_1$ at $u$ and $A_2$ at $y$, the min will be correct as long as all three cases lead to $(A_1, A_2)$-compatible paths, so that's what we'll prove.

    In the first case, assume that $v \neq y, u \neq z$. From the correctness of $cost[v, z]$, and the fact that $(u, v)$ and $(y, z)$ share no endpoints, it immediately follows that the $u \to t$ and $y \to t$ paths are $(A_1, A_2)$-compatible.

    In the second case, assume that $v = y$ (this implies that $z \neq u$, as otherwise $u$ and $y$ form a cycle). In other words, $A_1$ will go from $u$ to $y$ and meet $A_2$ there. Thus, we simply add the edge $(u, y)$ to $A_1$'s path and continue the traversal with both agents at $y$. So by the correctness of $cost[y, y]$, it follows that the $u \to t$ and $y \to t$ paths are $(A_1, A_2)$-compatible. 

    The third case, where $z = u$, is symmetric to the second case, but with the agents swapped.
\end{proof}
Suppose that there is a local budget of $k$ chunks per edge.
\begin{restatable}{theorem}{twoagentslocal}
    Given any task graph $G = (V, E)$ and a local constraint $k$, we can optimally chunk $G$ for two types of agents in time $O(|E|^2k^2 + |V|)$.
\end{restatable}
\begin{proof}[Proof Sketch]
    The runtime of the algorithm is dominated by determining when it's possible to split the agents onto separate paths. All together, this will take $O(|E|^2)$ applications of the algorithm in Section 6.1, for a total runtime of $O(|E|^2k^2)$. The algorithm first computes $\mathcal{P}(u, y)$ for all $u, y \in V$, and then computes the $cost$ recurrence.
    More details can be found in the appendix.
\end{proof}
Finally, suppose there is a global budget of $k$ chunks.
\begin{restatable}{theorem}{twoagentsglobal}
    Given any task graph $G = (V, E)$ and a global constraint $k$, we can optimally chunk $G$ for two types of agents in time $O(|E|^2 k^3 \log k + |V|)$.
\end{restatable}
\begin{proof}[Proof Sketch]
    Like in the single-agent global budget case, we first modify the function $\mathcal{P}$ to $\mathcal{P}'$, where $\mathcal{P'}(u, y)$ returns the set of $(v, z, i)$ such that $i$ is the minimum number of chunks to compatibly chunk $(u, v)$ and $(y, z)$ (where $i = \infty$ if no chunking is possible). The bottleneck is in computing the minimum number of chunks to split the agents from one vertex to two separate vertices. 
\end{proof}
\subsubsection{$m$ Types of Agents Taking the Same Path}
Assume that there are $m$ types of agents but only chunkings
where all $m$ types take the same path are allowed. This easily
reduces to the single agent case (found in Section 4), but we simply
use \autoref{alg:greedy2} to determine what edges we can persuade the
group of agents to take. More detail can be found in the appendix;
here, we simply state the main theorems. 
\begin{restatable}{theorem}{mAgentsLocal}
    Given any task graph $G = (V, E)$ and a local constraint $k$, we can find the optimal single-path chunking of $G$ for $m$ types of agents with at most $|E|$ applications of \autoref{alg:greedy2}, for a total runtime of $O(|E|mk + |V|)$.
\end{restatable}
\begin{restatable}{theorem}{mAgentsGlobal}
    Given any task graph $G = (V, E)$ and a global constraint $k$, we can find the optimal single-path chunking of 
 $G$ for $m$ types of agents with at most $|E|\log k$ applications of \autoref{alg:greedy2}, for a total runtime of $O(|E|mk\log k + |V|)$.
\end{restatable}


\section{Conclusion}

We have supplemented a graph-theoretic model of present bias
with a model of chunking, giving task designers the ability to chunk
edges in order to reduce the impact of present bias. We found that the
best way to chunk an edge is relatively straightforward for edges
on the shortest path, but significantly more complicated for edges
off the shortest path. We then used our optimal edge-chunking
algorithm to optimally chunk task graphs. We provided tight
theoretical guarantees on how much we can reduce an agent's cost ratio
as a function of the number of chunks we place in the graph. Finally, we showed how to optimally chunk task graphs for two types of agents simultaneously. Overall,
our work highlights the efficacy of chunking as a means to defeat the
harms agents incur due to their present bias.  

Our work raises several open questions. We highlight two interesting future directions. 
First, we saw that the problem grew significantly more complicated when designing for two types of agents. Can we extend our results to an arbitrary number of types? More generally, suppose the task designer was uncertain about the agents' present-bias and captured this uncertainty with a distribution over $b$. Our work can be seen as solving this problem when the support of this bias distribution is two. But can we chunk in the case where $b$ is continuously distributed? 

Second, as explained before, our model is best understood as the task designer suggesting a chunking to agents, rather than enforcing this chunking. In some situations, such as classroom settings, the task designer may want to place intermediate checkpoints to guarantee that agents make regular progress on the task. How should these checkpoints be modeled, and how much can they lower agents' costs compared to chunking?

%
%
%
%
%

\section*{Acknowledgements}
  The authors were supported in part by NSF grant
 IIS-1703846, MURI grant W911NF-19-1-0217, ARO grant W911NF-22-1-0061,
 and AFOSR grant FA23862114029.

\bibliographystyle{ACM-Reference-Format}
\bibliography{refs}


\begin{thebibliography}{15}


\ifx \showCODEN    \undefined \def \showCODEN     #1{\unskip}     \fi
\ifx \showDOI      \undefined \def \showDOI       #1{#1}\fi
\ifx \showISBNx    \undefined \def \showISBNx     #1{\unskip}     \fi
\ifx \showISBNxiii \undefined \def \showISBNxiii  #1{\unskip}     \fi
\ifx \showISSN     \undefined \def \showISSN      #1{\unskip}     \fi
\ifx \showLCCN     \undefined \def \showLCCN      #1{\unskip}     \fi
\ifx \shownote     \undefined \def \shownote      #1{#1}          \fi
\ifx \showarticletitle \undefined \def \showarticletitle #1{#1}   \fi
\ifx \showURL      \undefined \def \showURL       {\relax}        \fi
\providecommand\bibfield[2]{#2}
\providecommand\bibinfo[2]{#2}
\providecommand\natexlab[1]{#1}
\providecommand\showeprint[2][]{arXiv:#2}

\bibitem[Akerlof(1991)]%
        {akerlof1991procrastination}
\bibfield{author}{\bibinfo{person}{George~A Akerlof}.}
  \bibinfo{year}{1991}\natexlab{}.
\newblock \showarticletitle{Procrastination and obedience}.
\newblock \bibinfo{journal}{\emph{The american economic review}}
  \bibinfo{volume}{81}, \bibinfo{number}{2} (\bibinfo{year}{1991}),
  \bibinfo{pages}{1--19}.
\newblock


\bibitem[Albers and Kraft(2017)]%
        {albers2017value}
\bibfield{author}{\bibinfo{person}{Susanne Albers} {and}
  \bibinfo{person}{Dennis Kraft}.} \bibinfo{year}{2017}\natexlab{}.
\newblock \showarticletitle{On the value of penalties in time-inconsistent
  planning}.
\newblock \bibinfo{journal}{\emph{arXiv preprint arXiv:1702.01677}}
  (\bibinfo{year}{2017}).
\newblock


\bibitem[Albers and Kraft(2019)]%
        {albers2019motivating}
\bibfield{author}{\bibinfo{person}{Susanne Albers} {and}
  \bibinfo{person}{Dennis Kraft}.} \bibinfo{year}{2019}\natexlab{}.
\newblock \showarticletitle{Motivating time-inconsistent agents: A
  computational approach}.
\newblock \bibinfo{journal}{\emph{Theory of computing systems}}
  \bibinfo{volume}{63}, \bibinfo{number}{3} (\bibinfo{year}{2019}),
  \bibinfo{pages}{466--487}.
\newblock


\bibitem[Anagnostopoulos et~al\mbox{.}(2020)]%
        {anagnostopoulos2020collaborative}
\bibfield{author}{\bibinfo{person}{Aris Anagnostopoulos},
  \bibinfo{person}{Aristides Gionis}, {and} \bibinfo{person}{Nikos
  Parotsidis}.} \bibinfo{year}{2020}\natexlab{}.
\newblock \showarticletitle{Collaborative Procrastination}. In
  \bibinfo{booktitle}{\emph{10th International Conference on Fun with
  Algorithms (FUN 2021)}}. Schloss Dagstuhl-Leibniz-Zentrum f{\"u}r Informatik.
\newblock


\bibitem[Cormen et~al\mbox{.}(2009)]%
        {cormenAlgorithms}
\bibfield{author}{\bibinfo{person}{Thomas~H. Cormen},
  \bibinfo{person}{Charles~E. Leiserson}, \bibinfo{person}{Ronald~L. Rivest},
  {and} \bibinfo{person}{Clifford Stein}.} \bibinfo{year}{2009}\natexlab{}.
\newblock \bibinfo{booktitle}{\emph{Introduction to Algorithms, Third Edition}
  (\bibinfo{edition}{3rd} ed.)}.
\newblock \bibinfo{publisher}{The MIT Press}.
\newblock
\showISBNx{0262033844}


\bibitem[Dementiev et~al\mbox{.}(2021)]%
        {dementiev2021inconsistent}
\bibfield{author}{\bibinfo{person}{Yuriy Dementiev}, \bibinfo{person}{Fedor~V
  Fomin}, {and} \bibinfo{person}{Artur Ignatiev}.}
  \bibinfo{year}{2021}\natexlab{}.
\newblock \showarticletitle{Inconsistent Planning: When in doubt, toss a coin!}
\newblock \bibinfo{journal}{\emph{arXiv preprint arXiv:2112.03329}}
  (\bibinfo{year}{2021}).
\newblock


\bibitem[Fomin et~al\mbox{.}(2020)]%
        {fomin2020present}
\bibfield{author}{\bibinfo{person}{Fedor~V Fomin}, \bibinfo{person}{Pierre
  Fraigniaud}, {and} \bibinfo{person}{Petr~A Golovach}.}
  \bibinfo{year}{2020}\natexlab{}.
\newblock \showarticletitle{Present-Biased Optimization}.
\newblock \bibinfo{journal}{\emph{arXiv preprint arXiv:2012.14736}}
  (\bibinfo{year}{2020}).
\newblock


\bibitem[Gravin et~al\mbox{.}(2016)]%
        {gravin2016procrastination}
\bibfield{author}{\bibinfo{person}{Nick Gravin}, \bibinfo{person}{Nicole
  Immorlica}, \bibinfo{person}{Brendan Lucier}, {and}
  \bibinfo{person}{Emmanouil Pountourakis}.} \bibinfo{year}{2016}\natexlab{}.
\newblock \showarticletitle{Procrastination with variable present bias}.
\newblock \bibinfo{journal}{\emph{arXiv preprint arXiv:1606.03062}}
  (\bibinfo{year}{2016}).
\newblock


\bibitem[Kleinberg and Oren(2014)]%
        {kleinberg2014naive}
\bibfield{author}{\bibinfo{person}{Jon Kleinberg} {and} \bibinfo{person}{Sigal
  Oren}.} \bibinfo{year}{2014}\natexlab{}.
\newblock \showarticletitle{Time-inconsistent planning: a computational problem
  in behavioral economics}. In \bibinfo{booktitle}{\emph{Proceedings of the
  fifteenth ACM conference on Economics and computation}}.
  \bibinfo{pages}{547--564}.
\newblock


\bibitem[Kleinberg et~al\mbox{.}(2016)]%
        {kleinberg2016sophisticated}
\bibfield{author}{\bibinfo{person}{Jon Kleinberg}, \bibinfo{person}{Sigal
  Oren}, {and} \bibinfo{person}{Manish Raghavan}.}
  \bibinfo{year}{2016}\natexlab{}.
\newblock \showarticletitle{Planning problems for sophisticated agents with
  present bias}. In \bibinfo{booktitle}{\emph{Proceedings of the 2016 ACM
  Conference on Economics and Computation}}. \bibinfo{pages}{343--360}.
\newblock


\bibitem[Ma et~al\mbox{.}(2019)]%
        {ma2019penalty}
\bibfield{author}{\bibinfo{person}{Hongyao Ma}, \bibinfo{person}{Reshef Meir},
  \bibinfo{person}{David~C Parkes}, {and} \bibinfo{person}{Elena Wu-Yan}.}
  \bibinfo{year}{2019}\natexlab{}.
\newblock \showarticletitle{Penalty Bidding Mechanisms for Allocating Resources
  and Overcoming Present Bias}.
\newblock \bibinfo{journal}{\emph{arXiv preprint arXiv:1906.09713}}
  (\bibinfo{year}{2019}).
\newblock


\bibitem[Oren and Soker(2019)]%
        {oren2019principal}
\bibfield{author}{\bibinfo{person}{Sigal Oren} {and} \bibinfo{person}{Dolav
  Soker}.} \bibinfo{year}{2019}\natexlab{}.
\newblock \showarticletitle{Principal-Agent Problems with Present-Biased
  Agents}. In \bibinfo{booktitle}{\emph{International Symposium on Algorithmic
  Game Theory}}. Springer, \bibinfo{pages}{237--251}.
\newblock


\bibitem[Roughgarden(2016)]%
        {roughgarden2016cs269i}
\bibfield{author}{\bibinfo{person}{Tim Roughgarden}.}
  \bibinfo{year}{2016}\natexlab{}.
\newblock \showarticletitle{CS269I: Incentives in Computer Science Lecture\#19:
  Time-Inconsistent Planning}.
\newblock  (\bibinfo{year}{2016}).
\newblock


\bibitem[Saraf et~al\mbox{.}(2020)]%
        {saraf2020competition}
\bibfield{author}{\bibinfo{person}{Aditya Saraf}, \bibinfo{person}{Anna~R
  Karlin}, {and} \bibinfo{person}{Jamie Morgenstern}.}
  \bibinfo{year}{2020}\natexlab{}.
\newblock \showarticletitle{Competition Alleviates Present Bias in Task
  Completion}. In \bibinfo{booktitle}{\emph{International Conference on Web and
  Internet Economics}}. Springer, \bibinfo{pages}{266--279}.
\newblock


\bibitem[Tang et~al\mbox{.}(2017)]%
        {tang2017computational}
\bibfield{author}{\bibinfo{person}{Pingzhong Tang}, \bibinfo{person}{Yifeng
  Teng}, \bibinfo{person}{Zihe Wang}, \bibinfo{person}{Shenke Xiao}, {and}
  \bibinfo{person}{Yichong Xu}.} \bibinfo{year}{2017}\natexlab{}.
\newblock \showarticletitle{Computational issues in time-inconsistent
  planning}. In \bibinfo{booktitle}{\emph{Thirty-First AAAI Conference on
  Artificial Intelligence}}.
\newblock


\end{thebibliography}

\appendix
\allowdisplaybreaks

\section{Optimal Edge Chunking Proofs}
\chunkOptEdge*
\begin{proof}
  \autoref{lem:base_case} proves the case where $k = 2$.
Suppose that the theorem holds for $k-1$ chunks; we prove it for
  $k$ 
    chunks. For now, we assume that the shortest path from $u_i$ to
    $u$ is $u_{i+1}$ for all $i > 1$. At the end, we'll consider when
    this is not true. Say we put cost $x_1$ on the first edge. Then,
    we apply the inductive hypothesis to the other $k-1$ edges, now
    with a task of cost $x - x_1$. The costs $x_2, \dots, x_k$ are
    thus: 
    \begin{align*}
            x_i = \frac{(b-1)^{k-i-1}b^{i-2}}{b^{k-1}-(b-1)^{k-1}}(x-x_1).
    \end{align*}
    Because the shortest path through $u_i$ is $u_{i+1}$ for all $i > 2$ as well, we know from the inductive hypothesis that this chunking is optimal (given that $x_1$ is on the first edge). Further, the perceived costs of starting with those edges are all
    \begin{align*}
     \frac{1}{1 - \left(\frac{b-1}{b}\right)^{k-1}} (x - x_1) + c(v \to t).
    \end{align*}
    We want to minimize the maximum of the perceived cost of starting
    with edge $(u_1, u_2)$ and all the other edges. As before, we can
    do so by setting the perceived costs equal, as one side is
    decreasing in $x_1$ while the other is increasing in
    $x_1$. Because the shortest path from $u_i$ is through $u_{i+1}$
    for $i > 1$, the perceived cost of starting with $(u_1, u_2)$ is
    $bx_1 + c(u_2 \to t) = bx_1 + x_2 + c(u_3 \to t) = \dots = bx_1 +
    \sum_{i = 2}^k x_i + c(v \to t) = bx_1 + x - x_1 + c(v \to t) =
    (b-1)x_1 + x + c(v \to t)$. 
    \begin{align*}
        (b-1)x_1 + x + c(v \to t) &= \frac{1}{1 - \left(\frac{b-1}{b}\right)^{k-1}} (x - x_1) + c(v \to t) \\
        (b-1)x_1 + x &= \frac{b^{k-1}}{b^{k-1} - (b-1)^{k-1}} (x - x_1) \\
        \left(b-1 + \frac{b^{k-1}}{b^{k-1} - (b-1)^{k-1}}\right) x_1 &= \left(\frac{b^{k-1}}{b^{k-1} - (b-1)^{k-1}} - 1\right)x \\
        \frac{(b-1)(b^{k-1} - (b-1)^{k-1})+b^{k-1}}{b^{k-1} - (b-1)^{k-1}} x_1 &= \frac{b^{k-1} - b^{k-1} + (b-1)^{k-1}}{b^{k-1} - (b-1)^{k-1}} x \\
        (b^{k} - b^{k-1} - (b-1)^{k} + b^{k-1})x_1 &= (b-1)^{k-1} x \\
                x_1 &= \frac{(b-1)^{k-1}}{b^{k} - (b-1)^{k}} x.
    \end{align*}
    Thus, $x_1$ matches the chunking in the theorem. We now verify that the perceived cost matches:
    \begin{align*}
        (b-1)x_1 &+ x + c(v \to t) = (b-1)\cdot\frac{(b-1)^{k-1}}{b^{k} - (b-1)^{k}} \cdot x + x + c(v \to t) \\
        &= \left(\frac{(b-1)^{k}}{b^{k} - (b-1)^{k}} + 1\right) x + c(v \to t) \\
        &= \frac{(b-1)^{k} + b^{k} - (b-1)^{k}}{b^{k} - (b-1)^{k}}  x + c(v \to t) \\
        &= \frac{b^{k}}{b^{k} - (b-1)^{k}}x + c(v \to t)\\
            &= \frac{1}{1 - \left(\frac{b-1}{b}\right)^{k}}x + c(v \to t).
    \end{align*}
    A similar calculation will show that all the perceived costs are the same:
    \begin{align*}
        &\frac{1}{1 - \left(\frac{b-1}{b}\right)^{k-1}} (x - x_1) + c(v \to t) \\
        &= \frac{b^{k-1}}{b^{k-1} - (b-1)^{k-1}}\left(x - \frac{(b-1)^{k-1}}{b^{k} - (b-1)^{k}} x\right) + c(v \to t) \\
        &= \frac{b^{k-1}}{b^{k-1} - (b-1)^{k-1}}\left(1 - \frac{(b-1)^{k-1}}{b^{k} - (b-1)^{k}} \right)x + c(v \to t) \\
        &= \frac{b^{k-1}}{b^{k-1} - (b-1)^{k-1}}\left(\frac{b^{k} - (b-1)^{k} - (b-1)^{k-1}}{b^{k} - (b-1)^{k}} \right)x + c(v \to t) \\
        &= \frac{b^{k-1}}{b^{k-1} - (b-1)^{k-1}}\left(\frac{b^{k} - (b-1)^{k-1}(b-1+1)}{b^{k} - (b-1)^{k}} \right)x + c(v \to t) \\
        &= \frac{b^{k-1}}{b^{k-1} - (b-1)^{k-1}}\left(\frac{b(b^{k-1} - (b-1)^{k-1})}{b^{k} - (b-1)^{k}} \right)x + c(v \to t) \\
        &= \frac{b^{k}}{b^{k} - (b-1)^{k}}x + c(v \to t) \\
            &= \frac{1}{1 - \left(\frac{b-1}{b}\right)^{k}}x + c(v \to t).
    \end{align*}
    Finally, we can plug the value of $x_1$ into the formula for $x_i$:
    \begin{align*}
        x_i &= \frac{(b-1)^{k-i}b^{i-2}}{b^{k-1}-(b-1)^{k-1}}(x-x_1) \\
         &= \frac{(b-1)^{k-i}b^{i-2}}{b^{k-1}-(b-1)^{k-1}}\left(x-\frac{(b-1)^{k-1}}{b^{k} - (b-1)^{k}} x\right) \\
         &= \frac{(b-1)^{k-i}b^{i-2}}{b^{k-1}-(b-1)^{k-1}}\left(1-\frac{(b-1)^{k-1}}{b^{k} - (b-1)^{k}} \right) x\\
        &=
        \frac{(b-1)^{k-i}b^{i-2}}{b^{k-1}-(b-1)^{k-1}}\left(\frac{b(b^{k-1}
                    - (b-1)^{k-1})}{b^{k} - (b-1)^{k}} \right) x \tag{see
          previous derivation}\\ 
                &= \frac{(b-1)^{k-i}b^{i-1}}{b^{k} - (b-1)^{k}} x. 
    \end{align*}
    Thus, we've shown all components of the inductive statement. To summarize, under the assumption that the shortest path from $u_i$ is through $u_{i+1}$ for all $i>1$, we've shown that the chunking in the theorem is optimal and produces the correct selective bias.

    When the shortest path from $u_i$ to $t$ is through some external vertex $w$ instead of $u_{i+1}$, we've overestimated the perceived cost at some edges. In our calculations, we assumed that all edges would have perceived cost $bx_i + \sum_{j > i} x_j + c(v \to t)$, but actually some edges would have a \textit{lower} perceived cost of $p(e_i) = bx_i + c(w \to t)$. However, the final edge $(u_k, v)$ would still have perceived cost $bx_k + c(v \to t)$, as we assumed in the theorem, and thus the perceived cost of that edge in the chunking would be $\frac{1}{1 - \left(\frac{b-1}{b}\right)^{k}}x + c(v \to t)$. So, though optimality can no longer be guaranteed, the chunking in the theorem produces the expected selective bias regardless of whether the edge is on the shortest path.
\end{proof}
As a brief sanity check, we show that the $x_i$'s defined in the theorem actually sum to $x$.
\begin{proposition}
    $$\forall k \ge 1, \sum_{i=1}^k (b-1)^{k-i}b^{i-1} = b^{k} -
    (b-1)^{k}.$$ 
\end{proposition}
\begin{proof}
    When $k = 1$, the left side is $(b-1)^0 b^0 = 1$, while the right
        side is $b^1 - (b-1)^1 = 1$. Suppose that the statement holds for
    $k$. Then: 
    \begin{align*}
        \sum_{i=1}^{k+1} (b-1)^{k+1-i}b^{i-1} &= b^{k} + \sum_{i=1}^{k} (b-1)^{k+1-i}b^{i-1} \\
        &= b^{k} + (b-1)\sum_{i=1}^{k} (b-1)^{k-i}b^{i-1} \\
        &= b^{k} + (b-1)(b^{k} - (b-1)^{k}) \tag*{by the inductive hypothesis} \\
        &= b^{k} + b^{k+1} - b^{k} - (b-1)^{k+1} \\
                &=  b^{k+1} - (b-1)^{k+1}.
    \end{align*}
\end{proof}

\chunkImprove*
\begin{proof}
    Let $C$ have bottleneck $\beta$ and $O$ have bottleneck $\beta'$, where both chunkings have the same transition vertex $\tau$. Let $J = \{j: p(e_j^C) < \beta\}$ and $I = \{i: p(e_i^C) = \beta\}$ partition the indices. We will show that $\sum_{j \in J} x_j^O > \sum_{j \in J} x_j^C$ and that $x_i^O < x_i^C$ for all $i \in I$.
    
    Since $O$ has a lower bottleneck, it must be the case that $p(e_k^O) < \beta$ for all $k$. This implies that for all $i \in I$, we get that $p(e_i^O) < p(e_i^C)$ (since $p(e_i^C) = \beta$).
Note that $c(u_i^O \to t) = c(u_i^C \to t)$, as both
    chunkings have the same transition vertex $\tau$. Since $p(e_i^O)
    = bx_i^O + c(u_i^O \to t)$ and $p(e_i^C) = bx_i^C + c(u_i^C \to
    t)$, the fact that $p(e_i^O) < p(e_i^C)$ implies that $x_i^O <
    x_i^C$. 

  Clearly if $x_i^O < x_i^C$ for all $i\in I$, then $\sum_{j
    \in J} x_j^O > \sum_{j \in J} x_j^C$, as $I$ and $J$ partition
    the indices, and both chunkings must sum to $x$. 
\end{proof}

\chunkingOpt*
\begin{proof}
    Let $C$ have bottleneck $\beta$ and transition vertex $\tau$, and let $O$ have bottleneck $\beta' < \beta$ (and an arbitrary transition vertex). We prove that $\sum_{i = 1}^j x_i^C > \sum_{i=1}^j x_i^O$ for all $j$ by induction. With this proven, we get our desired contradiction with $\sum_{i = 1}^k x_i^C = x > \sum_{i=1}^k x_i^O$, which means that $O$ does not assign all the cost.

    For the base case of $j = 1$, note that $p(e_1^C) > p(e_1^O)$ (because the bottleneck is lower). Expanding the perceived cost equations:
    \begin{align*}
        p(e_1^C) &> p(e_1^O)  \\
        bx_1^C + c(u_2^C \to t) &> bx_1^O + c(u_2^O \to t) \\
        bx_1^C + \min{c(u,w) + c(w \to t), x - x_1^C + c(v \to t)} &>
                bx_1^O + \min{c(u,w) + c(w \to t), x - x_1^O + c(v \to t)}.  
    \end{align*}
If $x_1^O = x_1^C + \eps$ for any positive $\eps$, the first
    term would go up by $b\eps$ and the $\min$ would decrease by at
    most $\eps$ (if the both $\min$s were the second term). Because $b
    > 1$, this would never satisfy the above equation, and so $x_1^C >
    x_1^O$.  

    The inductive case is essentially analogous to the base case. The perceived cost equation for arbitrary $j$ expands to:
    \begin{align*}
        bx_1^C &+ \min(c(u,w) + c(w \to t), x - \sum_{i = 1}^{j-1} x_i^C - x_j^C + c(v \to t)) \\
        &> bx_1^O + \min(c(u,w) + c(w \to t), x - \sum_{i = 1}^{j-1}
                x_i^O - x_j^O + c(v \to t)). 
    \end{align*}
    The inductive hypothesis tells us that $\sum_{i = 1}^{j-1} x_i^C > \sum_{i = 1}^{j-1} x_i^O$, so these terms do not change the argument. The only way that the inequality can be satisfied is if $x_j^C < x_j^O$. Otherwise, if $x_1^O = x_1^C + \eps$ for any positive $\eps$, the first term would go up by $b\eps$ and the $\min$ would decrease by at most $\eps$ (since the sum is greater on the left hand side). So by induction, we get the desired result.
\end{proof}

\algChunkingOpt*
\begin{proof}
    Let $w$ denote the node following $u$ on the shortest path from $u$ to $t$. If $v=w$, we can simply apply \autoref{thm:chunk-opt-edge} to immediately get the best partition. So assume $v \neq w$. This means that $\delta >0$. 
    
    We first focus on the difficult case where $\delta \le x$; the
    case where $\delta > x$ will be covered at the end. As
    mentioned earlier, this means that we can satisfy any value of
    $\tau$, by placing at least $\delta$ cost on the first $\tau$
    edges while ensuring that the total cost of the first $\tau-1$ edges
    is less than $\delta$. The case where $\tau=k$ is an edge case that
        will be handled at the end. So suppose that $\tau \in \{1, \dots,
    k-1\}$. We explain how to optimally chunk $(u,
    v)$ for this fixed value of $\tau$; in other words, we produce the
    optimal chunking over all chunkings that satisfy $\sum_{i=1}^\tau
    x_i \ge \delta$ and $\sum_{i=1}^{\tau-1} x_i \le \delta$.  

    We start by setting $x_1 = x_2 = \dots = x_\tau =
        \delta/\tau$. Then for all $i < \tau$, $p(e_i) = bx_i + c(u,w) +
    c(w \to t) = \frac{b\delta}{\tau} + c(u,w) + c(w \to t)$. Further: 
   $$\begin{array}{ll}
      p(e_\tau) &= bx_\tau + c(u_{\tau+1} \to t)\\ &= bx_\tau + \sum_{i
              = \tau+1}^k x_i + c(v \to t)  \hfill \mbox{(shortest path from
          $u_{\tau+1}$ follows the chunking)}\\ 
            &= bx_\tau + x - \sum_{i=1}^\tau x_i + c(v \to t)
\\                        &= b\cdot \frac{\delta}{\tau} - \tau\cdot\frac{\delta}{\tau} +
x + c(v \to t) \hfill \mbox{(substituing $x_{i} = \delta/\tau$ for $i \le \tau$)}\\
                    &= \frac{b\delta}{\tau} + c(u, w) + c(w \to t)
\hfill \mbox{(since $\delta = x + c(v \to t) - c(u,w) - c(w \to t)$).}   
\end{array}$$
    Let $\alpha = \frac{b\delta}{\tau} + c(u,w) + c(w \to t)$. Then $p(e_i) = \alpha$ for all $i \le \tau$.

    Now we can chunk the remaining $x - \delta$ cost over the remaining $k-\tau$ edges according to \autoref{thm:chunk-opt-edge}, which gives them perceived costs:
    \begin{align*}
        \frac{x-\delta}{1-\left(\frac{b-1}{b}\right)^{k-\tau}} + c(v \to t) \eqdef \beta.
    \end{align*}

    
    From \autoref{lm:chunking_opt}, we know that if $\alpha = \beta$, we have the optimal chunking (for \textit{any} transition vertex $\tau$, not just the current $\tau$). In that case, we stop the algorithm and return this chunking. Otherwise, there are two cases:

    Case 1: $\alpha > \beta$. In this case, we claim that our chunking is optimal among all chunkings with transition vertex $\tau$. Notice that our chunking has bottleneck $\alpha$. By \autoref{lm:chunking_improve}, if another chunking, $O$ with the same $\tau$ has bottleneck lower than $\alpha$, it must assign lower cost to all of the first $\tau$ edges. But this means that $\sum_{i=1}^\tau x_i^O < \delta$, which means the transition vertex would be later than $\tau$. Thus, if $\alpha > \beta$, our chunking is optimal (for this $\tau$).

    Case 2: $\beta > \alpha$. The key to this case is that the
    perceived cost of starting with $e_\tau$ can be understood in two
    ways, which allows us to group it into either the earlier or later
    set of edges. This isn't the case for any other edge, and using
    this fact will allow us to modify our original chunking to lower
    $\beta$. More specifically, the original chunking ensures that the
    perceived cost of starting with $e_\tau$ is equal to all previous
    edges; the first modification we do in this case is to set the
    perceived cost of starting with $e_\tau$ equal to all \textit{later}
    edges instead. 
    
    We start by leaving $x_i$ fixed at $\delta/\tau$ for all $i <
    \tau$, but then chunking the remaining $x - \delta\cdot
    \frac{\tau-1}{\tau}$ work over the remaining $k-1+1$ edges
    according to \autoref{thm:chunk-opt-edge}, which modifies
    $x_\tau$. Because this assignment equalizes the perceived cost of
    starting with $e_\tau$ with that of later edges, it must have increased
    $x_\tau$ to be higher than $\delta/\tau$; by similar reasoning,
    all $x_i$ where $i > \tau$ must have decreased. Thus, this
    chunking has $p(e_i) = \beta'$ for all $i \ge \tau$, where $\beta'
    < \beta$. Further, since the perceived cost of starting with
    $e_\tau$ was $\alpha$, and $x_\tau$ increased, the new perceived
    cost of starting with $e_\tau$, $\beta'$ must still be higher than
    $\alpha$.   

    We now increase $x_1, \dots, x_{\tau-1}$ to raise $\alpha$
    and lower $\beta'$. We do so by setting $x_{i}$ to a
    placeholder $y$ for all $i < \tau$ and then solving for the
    optimal $y$. Note setting all these values equal is (weakly)
    dominant, because the perceived costs of starting with these edges
    are all $bx_i + c(u,w) + c(w \to t)$. Thus, if another chunking
    had $x_i \neq x_j$, where $i, j < \tau$, then setting $x_i$ and
    $x_j$ equal to their average would only decrease $\max_{i < \tau}
    p(e_i)$. This would either reduce the bottleneck (if the
    bottleneck is before $\tau$) or keep it the same. So we can set
    them all equal to $y$ without loss of generality. 

    With this, $p(e_i)$ for $i < \tau$ is $by + c(u,w) + c(w \to
    t)$. We then use \autoref{thm:chunk-opt-edge} to optimally split
    the remaining $x - y(\tau-1)$ work over the remaining $k-\tau+1$
        edges. With that, for all $i \ge \tau$, we get
    \begin{align*}
        p(e_i) = \frac{x-y(\tau-1)}{1 -
                    \left(\frac{b-1}{b}\right)^{k+1-\tau}} + c(v \to t). 
    \end{align*}

    We now set the two perceived costs equal and solve for the best $y$:
    \begin{align*}
        \frac{x-y(\tau-1)}{1 - \left(\frac{b-1}{b}\right)^{k-\tau+1}}
                + c(v \to t) &= by + c(u, w) + c(w \to t) \mbox{, so}\\ 
        y\left(\frac{\tau-1}{1 -
          \left(\frac{b-1}{b}\right)^{k-\tau+1}} + b\right) &=
        \frac{x}{1 - \left(\frac{b-1}{b}\right)^{k-\tau+1}} + c(v \to
                t) - c(w \to t) - c(u, w). 
    \end{align*}
    To ease notation, let $z_\tau = 1 -
\left(\frac{b-1}{b}\right)^{k-\tau+1}$.  We can then simplify as  follows: 
    \begin{align*}
        y\left(\frac{\tau-1+z_\tau b}{z_\tau}\right) &=
                \frac{x}{z_\tau} + c(v \to t) - c(u,w) - c(w \to t),
                \mbox{ so} \\ 
        y &= \frac{x + z_\tau(c(v \to t) - c(w \to t) - c(u, w))}{\tau-1+z_\tau b} \\
        &= \frac{\delta z_\tau + (1-z_\tau)x}{\tau-1+z_\tau b} \eqdef y^*.
    \end{align*}
For our final chunking, $C^*$, we set $x_{i} = \min(y^*,
    \frac{\delta}{\tau-1})$ for $i < \tau$, and split the remaining
    work over the latter edges via \autoref{thm:chunk-opt-edge}. Under
    this chunking, let $\alpha^* = p(e_i)$ for $i < \tau$ and let
    $\beta^* = p(e_i)$ for $i \ge \tau$. We claim the following. 
    \begin{restatable}{claim}{finalAlpha}
        \label{cl:final_alpha}
        $\alpha^*, \beta^* > \alpha$
    \end{restatable}
    The intuition for this is that the chunking $C^*$ increases the
        cost of early edges, while decreasing the cost of later
    edges. But we still ensure that the later edges have perceived
    cost at least as great as the early edges.
    \begin{proof}
    Note that with $y = \delta/l$, we got that $\alpha <
    \beta'$. Further, with $y = y^*$, the perceived costs starting with
    any edge would be equal, by definition of $y^*$. Thus, we know
    that $y^* > \delta/l$. It follows that $\min(y^*,
    \frac{\delta}{l-1}) > \delta/l$, and thus $\alpha^* > \alpha$.  

    Note that if $y = y^*$, then $\beta^* = \alpha^* > \alpha$,
    since choosing edge costs so that the perceived costs of starting
    with all edge in the chunking are equal 
    means that the costs on the early edges increase. Further,
    $\beta^*$ is decreasing in $y$. Since $y = \min(y^*,
    \frac{\delta}{l-1}) \le y^*$, this implies that $\beta^* >
    \alpha$. 
    \end{proof}

    We now show that $C^*$ has transition vertex $\tau$. By
    construction, we have that $\sum_{i < \tau} x_i \le
    \delta$. Let $i \le \tau$ be arbitrary. When $x_i$ was
    $\delta/\tau$ in the original chunking, we had that $p(e_i)$ was
    $\alpha$. By \autoref{cl:final_alpha}, we know that $p(e_i) >
    \alpha$, which means that $x_i > \delta/\tau$ (since perceived
    costs are strictly increasing in the actual cost). Thus, $\sum_{i
      \le \tau} x_i \ge \delta$.  

    We claim that $C^*$ is optimal (for the fixed transition
    vertex). First, note that if $y^* \le \frac{\delta}{\tau-1}$ and
    thus $x_i = y^*$ for all $i < \tau$, then $\alpha^* = \beta^*$ and
    the chunking is optimal (over \textit{all} transition vertices) by
    \autoref{lm:chunking_opt}. Otherwise, suppose that $y^* >
    \frac{\delta}{\tau-1}$ and so $x_{i} = y = \frac{\delta}{\tau-1}$
    for all $i < \tau$. Since $y < y^*$, and $\alpha^*$ is increasing
        in $y$, we know that $\beta^* > \alpha^*$. So the bottleneck of
    $C^*$ is $\beta^*$ in this case; by \autoref{lm:chunking_improve},
    any better chunking $O$ with the same transition vertex must have
    $\sum_{i < \tau} x_i^O > \sum_{i < \tau} x_i^{C^*} =
    \delta$. Thus, $O$ would have an earlier transition vertex, which
    is a contradiction.  

    Lastly, we discuss the runtime of the algorithm. In our analysis,
    for a fixed $\tau$, we must compare the $\alpha$ and $\beta$
    values in two chunkings -- the initial one where $x_i = \delta/t$
    for all $i \le \tau$, and the modified one where $x_i = y$ for all
    $i < \tau$. Since we have closed-form equations for the $\alpha$
    and $\beta$ values in each chunking, we do not need to construct
    them for each $\tau$. We simply keep track of which value of
    $\tau$ produces the smallest perceived cost, and whether the best
    chunking for that $\tau$ was the initial chunking or the modified
    one. We can thus do only constant work for each $\tau$, resulting
    in a runtime of $O(k)$. See \autoref{alg:opt-edge-chunk} for
    details. 

    Finally, we prove the remaining two edges cases.

    The first is when $\tau=k$. In this case, the first chunking would set all costs equal to $\delta/k$, which would not cover the full cost of the original edge. However, this case is also very simple, as all edges have the same perceived cost of $bx_i + c(u,w) + c(w \to t)$ when $\tau=k$. So, this case proceeds as follows. First, we set all $x_i = x/k$. If $x/k < \frac{\delta}{k-1}$, this would satisfy the constraint that $\tau=k$, and since all edges would have the same perceived cost, this would be optimal. Otherwise, we would set $x_{i} = \frac{\delta}{k-1}$ for all $i < k$ and $x_k = x - \delta$, which would be optimal for $\tau=k$, as this would be as close as we could get to uniform costs. 
    
    Finally, we consider the case where $\delta > x$. We established
    earlier that the shortest path will switch from the chunking to
    the $w$ vertices if at least $\delta$ work has been completed on
    the chunking. Since $\delta > x$, this can't happen, and so no
    matter how we chunk, the shortest path from any $u_{<k}$ is
    through $w$. This means that $p(e_i) = bx_i + c(u, w) + c(w \to
    t)$ for all $i < k$. Note that $e_k = (u_k, v)$; so, this final
    edge locks the agent into going to $v$. Thus, $p(e_k) = bx_k + c(v
    \to t) = b(x - \sum_{i<k}x_i)+c(v \to t)$. To optimally chunk, we
    set all $x_i = y$ for $i < k$ and then set the perceived cost of
    starting with the final edge equal to this to find the optimal $y$. 
    \begin{align*}
        by + c(u,w) + c(w \to t) &=  b(x - (k-1)y)+c(v \to t) \\
        byk &= bx + c(v \to t) - c(u,w) - c(w \to t) = \delta + (b-1)x \\
                y &= \frac{\delta+(b-1)x}{bk} \eqdef y^*.
    \end{align*}
We now simply set $y = \min\left(y^*, \frac{x}{k-1}\right)$. If
    $y^* \le \frac{x}{k-1}$, then all perceived costs are equal, so
    this chunking is optimal by \autoref{lm:chunking_opt}. If $y^* >
    \frac{x}{k-1}$, then the perceived cost of starting with the final
    edge is still higher, but the actual cost of that edge cannot be
    reduced below 0. Note that the case where $y^* > \frac{x}{k-1}$
    (and $\delta > x$) is the only case where the optimal chunking
    might put a cost of 0 on any edge. 
\end{proof}

\section{Cost Ratio Corollary}
\costRatioCor*
\begin{proof}
    Let $c$ be a constant. By \autoref{thm:costRatio}, we will get a cost ratio of $O(c)$ if $\bmin \le c^{1/n}$. We thus solve for the following equation for $k$:
    \begin{align*}
        \frac{1}{1-\left(\frac{b-1}{b}\right)^k} &= c^{1/n} \\
        \frac{1}{c^{1/n}} &= 1 - \left(\frac{b-1}{b}\right)^k \\
        \left(\frac{b-1}{b}\right)^k &= 1 - \frac{1}{c^{1/n}} \\
        \left(\frac{b-1}{b}\right)^k &= \frac{c^{1/n} - 1}{c^{1/n}} \\
        k &= \frac{\log\left(\frac{c^{1/n} - 1}{c^{1/n}}\right)}{\log\left(\frac{b-1}{b}\right)} \\
        &=
                \frac{\log\left(\frac{c^{1/n}}{c^{1/n}-1}\right)}{\log\left(\frac{b}{b-1}\right)}. 
    \end{align*}
Since $b$ is a constant, $\log\left(\frac{b}{b-1}\right)$ is constant,
and $k$ is thus dominated by the numerator. Similarly, $c^{1/n} < c$,
and thus we are interested in the asymptotic behavior of
$\log\left(\frac{1}{c^{1/n}-1}\right)$. The series expansion as $n \to
\infty$ is $\frac{n}{\log c} - \frac{1}{2} + \frac{\log c}{12n} +
O(\frac{1}{n^2}) = O(n)$. 
\end{proof}

\section{Non-short Path Edge Chunking Algorithm}
\DontPrintSemicolon

\setcounter{algocf}{0}
\begin{algorithm}[H]
    \scriptsize
    \KwIn{A DAG $G$, edge $(u, v)$ in $G$, bias factor $b$ and chunking parameter $k$}
    \KwOut{The optimal chunking for edge $(u, v)$ and the associated bottleneck cost}
    $x \gets c(u, v)$, 
    $w \gets $ next node in shortest $u \to t$ path\\
    \If(\tcp*[f]{edge case for when $(u, v)$ is on the shortest path}){$w = v$} {
        \Return $\textsf{Chunk-Shortest-Edge}(k, x), \frac{1}{1 - \left(\frac{b-1}{b}\right)^{k}} x + c(v \to t)$
    } 
    $\delta \gets x+c(v \to t) - c(u,w) -c(w \to t)$\\
    \If{$\delta > x$}{
        $y^* \gets \frac{\delta + (b-1)x}{bk}$ \\
        $C \gets x_1, \dots, x_{k-1} \mapsto \max(y^*, \frac{x}{k-1})$ and $x_k \mapsto x-(k-1)\max(y^*, \frac{x}{k-1})$ \\
        $min\_bottleneck \gets b\max(y^*, \frac{x}{k-1}) + c(u, w) + c(w \to t)$ \\
        \Return $C, min\_bottleneck$
    }
    $min\_bottleneck \gets \infty, \tau^* \gets 0, opt\_chunk\_type \gets 0$\\
    \For{$\tau = 1$ \KwTo $k-1$} {
        $\alpha_0 \gets \frac{b\delta}{\tau}+c(u,w)+c(w \to t)$, 
        $\beta_0 \gets \frac{x-\delta}{1-\left(\frac{b-1}{b}\right)^{k-\tau}}+c(v \to t)$\\
        \uIf{$\alpha_0 = \beta_0$}{
            $C \gets x_1, \dots, x_\tau \mapsto \delta/\tau$ and $x_{\tau+1}, \dots, x_k \mapsto \textsf{Chunk-Shortest-Edge}(k-\tau, x-\delta)$\\
            \Return $C, \alpha_0$
        }
        \uElseIf{$\alpha_0 > \beta_0$} {
            \If{$\alpha_0 < min\_bottleneck$}{
                $min\_bottleneck \gets \alpha_0$, 
                $\tau^* \gets \tau, opt\_chunk\_type \gets 0$
            }
        }
        \Else{
            \If(\tcp*[f]{edge case for $\tau = 1$}){$\tau = 1$} { 
                \Return $\textsf{Chunk-Shortest-Edge}(k, x), \frac{1}{1 - \left(\frac{b-1}{b}\right)^{k}} x + c(v \to t)$
            }
            $z_\tau \gets 1-\left(\frac{b-1}{b}\right)^{k-\tau+1}$,  
            $y^* \gets \frac{\delta z_\tau + (1-z_\tau)x}{\tau-1+z_\tau b}$ \\
            \eIf{$\frac{\delta}{\tau-1} > y^*$} {
                $C \gets x_1, \dots, x_{\tau-1} \mapsto y^*$ and $x_\tau, \dots, x_k \mapsto \textsf{Chunk-Shortest-Edge}(k-\tau+1, x-(\tau-1)y)$\\
                \Return $C, by^* + c(u,w) + c(w \to t)$
            } {
                $\beta \gets \frac{x-\delta}{z_\tau} + c(v \to t)$ \\
                \If{$\beta < min\_bottleneck$} {
                    $min\_bottleneck \gets \beta$, 
                    $\tau^* \gets \tau, opt\_chunk\_type \gets 1$
                }
            }
        }
    }
    \eIf{$\frac{x}{k} \le \frac{\delta}{k-1}$}{
        $C \gets x_1, \dots, x_k \mapsto x/k$, 
        $min\_bottleneck \gets \frac{bx}{k} + c(u, w) + d(w)$\\ 
        \Return $C, min\_bottleneck$
    }{
        $\alpha \gets \frac{b\delta}{k-1} + c(u, w) + c(w \to t)$, 
        $\beta \gets b(x - \delta) + c(v \to t)$ \\
        \If{$\min(\alpha, \beta) \le min\_bottleneck$}{
            $C \gets x_1, \dots, x_{k-1} \mapsto \frac{\delta}{k-1}$ and $x_k \mapsto x - \delta$, 
            $min\_bottleneck \gets \min(\alpha, \beta)$ \\
            \Return $C, min\_bottleneck$
        }
    }

    \eIf{$opt\_chunk\_type = 0$}{
        $C \gets x_1, \dots, x_{\tau^*} \mapsto \delta/{\tau^*}$ and $x_{{\tau^*}+1}, \dots, x_k \mapsto \textsf{Chunk-Shortest-Edge}(k-{\tau^*}, x-\delta)$
    }{
        $C \gets x_1, \dots, x_{\tau^*-1} \mapsto \frac{\delta}{\tau^*-1}$ and $x_{\tau^*}, \dots, x_k \mapsto \textsf{Chunk-Shortest-Edge}(k-\tau^*+1, x-\delta)$
    }
    \Return $C, min\_bottleneck$
    \caption{Optimally chunk any edge. Uses $\textsf{Chunk-Shortest-Edge}(k, x)$ as a subroutine, which returns the optimal $k$-chunking of a shortest edge of cost $x$, which is given by \autoref{thm:chunk-opt-edge}.}
    \label{alg:opt-edge-chunk}
\end{algorithm}

\section{Splitting Agents onto Separate Paths}
We first provide a full description of \autoref{alg:chunkSplit}.
\begin{algorithm}
    maxBottleneck $\gets 0$, $C^* \gets \emptyset$ \\
    \For{$i = 1$ \KwTo $k$} {
        $C_i = (x_1, \dots, x_k) \gets$ optimal chunking of $(u, v)$ for $A_1$\\
        \For{$j = i-1$ \KwTo $1$} {
            \If{$p(e_i; b_1) < \alpha_u^{(1)}$}{
               $\delta = \min(x_j, (p(e_i; b_1) - \alpha_u^{(1)})/b_1)$ \\
               $x_{j} \gets x_{j} - \delta$\\
               $x_{i} \gets x_i + \delta$
            }
        }
        \For{$j = i+1$ \KwTo $k$} {
            \If{$p(e_i; b_1) < \alpha_u^{(1)}$}{
                $\gamma_w \gets c(u, w) + c(w \to t), \gamma_x \gets \sum_{l > i} x_l + c(v \to t)$ \\
                \uIf{$\gamma_x \le \gamma_w$}{
                    $\delta = \min(x_{j}, (p(e_i; b_1) - \alpha_u^{(1)})/(b_1 - 1))$
                }
                \uElseIf{$p(e_i; b_1) - \alpha_u^{(1)} \le b_1(\gamma_x - \gamma_w)$}{
                    $\delta = \min(x_{j}, (p(e_i; b_1) - \alpha_u^{(1)})/(b_1 - 1))$
                }
                \uElse{
                    $\delta' = \gamma_x - \gamma_w$ \\
                    $x_j \gets x_j - \delta'$ \\
                    $x_i \gets x_i + \delta'$ \\
                    $\delta = \min(x_{j}, (p(e_i; b_1) - \alpha_u^{(1)})/(b_1 - 1))$
                }
                $x_j \gets x_j - \delta$\\
                $x_i \gets x_i + \delta$
            }
        }
        $\lambda \gets \sum_{j = 1}^{i-1} p(e_j; b_1) - \alpha_u^{(1)}$ \\
        $\delta^* \gets \min(\lambda/(b_1 - 1), \sum_{j > i}x_j/b_1$ \\
        decrease $x_{>i}$ by $b_1\delta^*$ \\
        $x_i \gets x_i + \delta^*$ \\
        $j \gets i-1$ \\
        \While{$\delta^* > 0$} {
            $\delta \gets \min(x_j, \delta^*)$ \\
            $x_j \gets x_j - \min(x_j, \delta)$ \\
            $\delta^* \gets \delta^* - \delta$\\
            $j \gets j - 1$
        }
        bottleneck $\gets \max_j p(e_j; b_2)$ \\
        \If{bottleneck > maxBottleneck}{
            maxBottleneck $\gets$ bottleneck \\
            $C^* \gets C$
        }
    }
    \Return $C^*$
\caption{Chunk $(u, v)$ such that $A_1$ takes the chunking and $A_2$
    doesn't, if possible.} 
    \label{alg:chunkSplit}
\end{algorithm}
We now prove that this algorithm is correct via the following theorem.
\chunkSplitThm*
\begin{proof}
  First, suppose that $\sum_{j \neq i} x_j = 0$, that is, all of the
  weight is on $e_i$ in $C_i$. It's obvious that $p(e'_i; b_2)
    \le p(e_i; b_2)$, as the perceived cost of any chunk cannot exceed
    $bc(u, v) + c(v \to t)$, and $C_i$ achieves this cost on $e_i$. 
    %
It follows that $\sum_{j \neq i} x_j > 0$, which implies that
    $p(e_i; b_1) = \alpha_u^{(1)}$ by
\autoref{lm:postconditions}(a). We now consider two cases. 
\paragraph{Case 1: $x_i' > x_i$.} Suppose that $\sum_{j > i} x_j =
    0$. We must have $x_j' \ge x_j$,for all $j \ge i$, as costs must be
non-negative. Thus, $p(e'_i; b_1) > p(e_i; b_1)$, and $A_1$ will
    deviate from $C'$ at edge $i$.  
    
    So, suppose instead that $\sum_{j > i} x_j > 0$. By
    \autoref{lm:postconditions}(b), we have that $\forall j \le i,
   p(e_j; b_1) = \alpha_u^{(1)}$. A similar now argument applies: if 
    $\sum_{j > i} x_j' < \sum_{j >i} x_j$, then more weight must be
    put on $x_{\le i}$, and it's clear that doing so would cause $A_1$
    to deviate before or at edge $i$ (concretely, $A_1$ would deviate
    at the first edge with higher weight). But if $\sum_{j > i} x_j'
    \ge \sum_{j > i} x_j$, then $p(e'_i; b_1) > p(e_i; b_1)$, and
    $A_1$ will deviate from $C'$ at edge $i$. Either way, $A_1$ will
    not take the chunking $C'$. 
    \paragraph{Case 2: $x_i' \le x_i$.} Recall that we can write $p(e_i; b_2)$ as $b_2 x_i + c(u_{i+1} \to t)$, where $c(u_{i+1} \to t)$, the cost of the cheapest path from $u_{i+1}$ to $t$, is $\min(c(u,w) + c(w \to t), \sum_{j>i} x_j + c(v \to t))$.
    \begin{align*} 
        p(e_i'; b_2) &> p(e_i; b_2) \\
        \iff b_2 x_i' + c(u'_{i+1} \to t) &> b_2 x_i + c(u_{i+1} \to t) \\
        \iff c(u'_{i+1} \to t) - c(u_{i+1} \to t) &> b_2 (x_i - x_i') \\
        \implies c(u'_{i+1} \to t) - c(u_{i+1} \to t) &> b_1 (x_i -
                x_i') \tag{since $b_2 > b_1$ and $x_i - x_i' \ge 0$}\\  
        \iff b_1 x_i' + c(u'_{i+1} \to t) &> b_1 x_i + c(u_{i+1} \to t) \\
                \iff p(e_i'; b_1) &> p(e_i; b_1).
    \end{align*}
    Since $p(e_i; b_1) = \alpha_u^{(1)}$, $A_1$ won't take $C'$ (they will deviate at $e_i'$).
\end{proof}

We now describe the flipped version of this problem, where we chunk $(u, v)$ so that $A_2$ takes it but $A_1$ finds it maximally unappealing. The flipped algorithm has the same phase 1 and 2 as before.\footnote{We omit the full pseudocode for the modified algorithm, as it's easy to modify the third phase of \autoref{alg:chunkSplit}.} Phase 3 is modified to:
\begin{enumerate}
    \item[3.] Let $\lambda$ be the total amount of cost that could be added to $x_{>i}$ while ensuring that $p(e_j; b_1) \le \alpha_u^{(1)}$ for all $j > i$. Let $\delta = \min(\lambda/b_2, \sum_{j < i} x_j/(b_2 - 1), x_i)$. Decrease $x_i$ by $\delta$, decrease the cumulative cost of $x_{<i}$ by $(b_2 - 1)\delta$, and increase the cumulative cost of $x_{>i}$ by $b_2\delta$.
\end{enumerate}
We also modify part (b) of the lemma.
\begin{lemma}\label{lm:postconditions2}
    Let $C = (e_1, \dots, e_k)$ be the chunking produced by the algorithm above. Then:
    \begin{enumerate}[(a)]
        \item $\sum_{j \neq i} x_j > 0 \implies p(e_i; b_2) = \alpha_u^{(2)}$
        \item $\sum_{j < i} x_j > 0$ and $x_i > 0 \implies \forall j > i, p(e_j; b_2) = \alpha_u^{(2)}$
    \end{enumerate}
\end{lemma}
\begin{proof}
    The proof of (a) is identical to before. For (b), as before, if more could be siphoned from $\sum_{j < i} x_j$ and $x_i$, the algorithm would, unless no edges in $e_{>i}$ can be increased further.
\end{proof}
\begin{theorem}
    \label{thm:diffPath}
    Let $C$ be the output of the algorithm above. Let $C'$ be another chunking such that $p(e_i'; b_1) > p(e_i; b_1)$. Then, $A_2$ will not take $C'$. 
\end{theorem}
\begin{proof}
    First, suppose that $\sum_{j \neq i} x_j = 0$, i.e., all of the weight is on $e_i$ in $C$. Then, it's obvious that $p(e'_i; b_1) \le p(e_i; b_1)$, as the perceived cost of any chunk cannot exceed $bc(u, v) + c(v \to t)$, and $C$ achieves this cost on $e_i$.

    So, we know that $\sum_{j \neq i} x_j > 0$, which implies that
    $p(e_i; b_1) = \alpha_u^{(1)}$ by
    \autoref{lm:postconditions2}(a). We now consider two cases. 
    \paragraph{Case 1: $\sum_{j > i} x_j' \le \sum_{j > i} x_j$.} Recall that $c(u_{i+1} \to t) = \min(c(u,w) + c(w \to t), \sum_{j>i} x_j + c(v \to t))$. Thus, $\sum_{j > i} x_j' \le \sum_{j > i} x_j$ implies that $c(u_{i+1} \to t) \ge c(u'_{i+1} \to t)$.
    \begin{align*}
         p(e_i'; b_1) &> p(e_i; b_1) \\
        \iff b_1 x_i' + c(u'_{i+1} \to t) &> b_1 x_i + c(u_{i+1} \to t) \\
        \iff b_1 (x_i' - x_i) &> c(u_{i+1} \to t) - c(u'_{i+1} \to t) \\
        \implies b_2 (x_i' - x_i) &> c(u_{i+1} \to t) - c(u'_{i+1} \to
                t) \tag{since $b_2 > b_1$ and $c(u_{i+1} \to t) - c(u'_{i+1}
          \to t) \ge 0$} \\ 
        \iff b_2 x_i' + c(u'_{i+1} \to t) &> b_2 x_i + c(u_{i+1} \to t) \\
                \iff p(e_i'; b_2) &> p(e_i; b_2).
    \end{align*}
    Since $p(e_i; b_2) = \alpha_u^{(2)}$, $A_2$ won't take $C'$.
    \paragraph{Case 2: $\sum_{j >i} x_j' > \sum_{j >i} x_j$.} Suppose,
    for the sake of contradiction, that $\sum_{j < i} x_j = 0$. Then,
    $\sum_{j \ge i} x_j = x$, i.e., all the weight is on edges $e_{\ge
      i}$. Now, $p(e_i';b_1) > p(e_i; b_1)$ requires either that $x_i'
    > x_i$, or that $C'$ assigns more cost to edges $e_{\ge i}$ than
    $C$. The latter is impossible because $C$ assigns all the weight
 to edges $e_{\ge i}$, and thus $x_i' > x_i$. This 
    implies that $\sum_{j>i} x_j' < \sum_{j >i} x_j$, which
gives us a contradiction.

So it follows that $\sum_{j < i} x_j > 0$. We now consider two
    cases. First, suppose that $x_i > 0$. We apply
    \autoref{lm:postconditions2}(b), which says that $\forall j > i,
    p(e_j; b_2) = \alpha_u^{(2)}$. Since $\sum_{j > i} x_j' > \sum_{j
      > i} x_j$, by \autoref{lm:m_agents_greedy} there must be some
    edge $e'_j$ such that $p(e'_j; b_2) > p(e_j; b_2) =
    \alpha_u^{(2)}$, and thus $A_2$ deviates from $C'$. 
    
    Second, suppose that $x_i = 0$. $x_i' \ge x_i$. This combined with the fact that $\sum_{j > i} x_j' > \sum_{j >i} x_j$ implies that $p(e_i'; b_2) > p(e_i; b_2) = \alpha_u^{(2)}$, and thus $A_2$ doesn't take $C'$. 
\end{proof}

\section{Keeping Agents on the Same Path}
\mAgentsGreedy*
\begin{proof}
    We prove the contrapositive. That is, suppose that, for all $i \in [l, k]$, there exists some $b_i$ such that $p(e_i; b_i) \ge p(e_i'; b_i)$. We show that $\sum_{i = j}^k x_i \ge \sum_{i = j}^k x_i'$ by induction from $j = k$ to $l$.

    For the base case, suppose that $j = k$. Note that $p(e_k; b_k)
    \ge p(e_k'; b_k)$ if and only if $b_k x_k + c(v \to t) \ge b_k
    x_k' + c(v \to t)$, which implies that $x_k \ge x_k'$, as
    desired. For the inductive case, assume that $\sum_{i > j}^k
    x_i \ge \sum_{i > j}^k x_i'$. First, we expand $p(e_j; b_j) \ge
    p(e_j'; b_j)$: 
    \begin{align}\label{eq:pcosts1}
        bx_j + \min(c(u, w) + c(w \to t), \sum_{i > j}^k x_i + c(v \to
        t)) &\ge bx_j' + \min(c(u, w) + c(w \to t), \sum_{i > j}^k
                x_i' + c(v \to t)). 
    \end{align}
        We now proceed by cases.
    \paragraph{Case 1.} Suppose that $\min(c(u, w) + c(w \to t), \sum_{i > 
      j}^k x_i' + c(v \to t)) = c(u, w) + c(w \to t)$. Since $\sum_{i
      > j}^k x_i' + c(v \to t) \le \sum_{i > j}^k x_i + c(v \to t)$ by
    the inductive hypothesis, we also know that $\min(c(u, w) + c(w
    \to t), \sum_{i > j}^k x_i + c(v \to t)) = c(u, w) + c(w \to
    t)$. Thus, \autoref{eq:pcosts1} holds if and only if: 
    \begin{align*}
        bx_j + c(u, w) + c(w \to t) &\ge bx_j' + c(u, w) + c(w \to t) \\
        \iff bx_j  &\ge bx_j' \\
                \iff x_j &\ge x_j'. 
    \end{align*}
    Combining this with the inductive hypothesis yields $\sum_{i \ge j} x_i \ge \sum_{i \ge j} x_i'$, as desired.
        \paragraph{Case 2.} Suppose that $\min(c(u, w) + c(w \to t), \sum_{i > 
      j}^k x_i' + c(v \to t)) = \sum_{i > j}^k x_i' + c(v \to
    t)$. Clearly $\min(c(u, w) + c(w \to t), \sum_{i    
      > j}^k x_i + c(v \to t)) \le \sum_{i > j}^k x_i + c(v \to
    t)$. Thus, \autoref{eq:pcosts1} implies: 
    \begin{align*}
        bx_j + \sum_{i > j} x_i + c(v \to t) &\ge bx_j' + \sum_{i > j} x_i' + c(v \to t)\\
        \iff bx_j + b\sum_{i > j} x_i &\ge bx_j' + \sum_{i > j} x_i' +
                (b-1)\sum_{i >j}x_i \tag{adding $(b-1)\sum_{i >j}x_i$ to both
          sides}\\ 
        \implies bx_j + b\sum_{i > j} x_i &\ge bx_j' + b\sum_{i > j}
                x_i'   \tag{since $\sum_{i >j}x_i \ge \sum_{i >j}x_i'$ by the
          IH}\\ 
                                \iff \sum_{i \ge j} x_i &\ge \sum_{i \ge j} x_i'.
    \end{align*}
    The last line proves the inductive step, and thus completes the proof.
\end{proof}

\section{Graph-Chunking Theorems for Multiple Agents}
\twoagentslocal*
\begin{proof}
The main computational bottleneck is computing $\mathcal{P}(u, y)$ for
all $u, y \in V$. For $u \neq y$, this is very simple: we can chunk
edges for each agent independently when they aren't at the same
node. Doing so requires $2|E|$ applications of
\autoref{alg:opt-edge-chunk} ($|E|$ applications for each agent), for
a runtime of $O(2|E|k)$. For $\mathcal{P}(u, u)$, consider all
$(v, z) \in N(u) \times N(u)$. There are a total of $|E|^2$ such pairs
over all choices of $u$. When $v \neq z$, we apply
\autoref{alg:chunkSplit} (to $(u, v)$ and $(u, z)$), and when $v = z$,
we apply \autoref{alg:greedy2}. \autoref{alg:chunkSplit} runs in
$O(k^2)$ time and \autoref{alg:greedy2} runs in $O(k)$ time (for $m =
2$ agents). Thus, the total runtime to compute $\mathcal{P}$ is
$O(2|E|^2k^2 + |E|k + 2|E|k) = O(|E|^2k^2)$. 

Once we have $\mathcal{P}$, we need to compute the cost recurrence. For each element in $\mathcal{P}$, we compute the min over the three constant time functions $C_1, C_2$, and $C_3$, for a total time of $O(|\mathcal{P}|) = O(|E|^2)$, since $\mathcal{P} \subseteq E \times E$. Thus, the cost recurrence takes $O(|E|^2)$ time to compute, which means that the total runtime is dominated by computing $\mathcal{P}$. 

For correctness, $\mathcal{P}$ is correct by \autoref{thm:diffPath}, \autoref{thm:chunk_split}, and \autoref{thm:samePath}. Given the correctness of $\mathcal{P}$, the cost recurrence is correct by \autoref{lm:costRec}.
\end{proof}

\twoagentsglobal*
\begin{proof}
    We slightly modify the definition of $\mathcal{P}'(u, y)$ to be the set of $(v, z, i)$ such that $i$ is the minimum number of chunks needed for $(u, v)$ and $(v, z)$ to be compatibly chunked. With this, we can modify the cost recurrence in the obvious way. The individual cases become:
    \begin{align*}
        C_1(u, v, y, i) &= \begin{cases*}
            c(y, u) + cost[u, u, i] & if $(v, u) \in \mathcal{P}'(u, y)$ \\
            \infty & otherwise
        \end{cases*} \\
        C_2(u, y, z, i) &= \begin{cases*}
            c(u, y) + cost[y, y, i] & if $(y, z) \in \mathcal{P}'(u, y)$ \\
            \infty & otherwise
        \end{cases*} \\
        C_3(u, v, y, z, i) &= c(u, v) + c(y, z) + cost[v, z, i].
    \end{align*} 
    And the recurrence becomes:
    \begin{align*}
        cost[u, y, i] = &\min_{(v, z, l) \in \mathcal{P}'(u, y): l \le
          i} \min(C_1(u, v, y, i-l), C_2(u, y, z, i-l), C_3(u, v, y,
                z, i-l)). 
    \end{align*}
    Computing this recurrence will take time $O(|E|^2k)$, but this will not be the bottleneck. The correctness of this recurrence follows simply from the correctness of $\mathcal{P}'$. It remains to show how to compute this new $\mathcal{P}'$.

    For $\mathcal{P}'(u, y)$ where $u \neq y$, it is easy to return
    the minimum number of chunks needed to chunk the edges; we already solved this problem
    with binary search in the single-agent global budget case
    (\autoref{thm:global-budget}). This takes $O(2|E|\log k)$ time for
    all $u \neq y$. Now suppose $u = y$. If $v = z$, and we're thus
    trying to keep agents on the same path, we can also use binary
    search with \autoref{alg:greedy2} to find the minimum number of
    chunks to get both agents to stick to the path. This takes
    $O(|E|\log k)$ time in total.  

    The bottleneck is computing the minimum number of chunks to get
    $A_1$ to take $(u, v)$ and $A_2$ to take $(u, z)$. We can
    visualize the problem as searching through a two dimensional
    binary array, where $arr[i, j]=1$ iff we can get a compatible
    chunking where $A_1$ takes an $i$-chunking of $(u, v)$ and $A_2$
    takes a $j$-chunking of $(u, z)$. Luckily, the array is row-wise
    and column-wise sorted; that is, we can always simulate an
    $i$-chunking with an $(i+1)$-chunking (e.g., set the first chunk
    to $0$), so if $A_1$ can take an $i$-chunking of $(u, v)$ and
    $A_2$ can take a $j$-chunking of $(u, z)$, then it's true that
    $A_1$ can take an $(i+1)$-chunking of $(u, v)$ and $A_2$ can take
    a $j$-chunking of $(u, z)$. Our goal is to find $\min_{i, j:
      arr[i, j] = 1} i + j$. In the worst case, the matrix has
        dimensions $k \times k$.\footnote{Technically, we care only about 
    the lower triangle (i.e., entries $arr[i, j]$ where $i + j \le
    k$), but this doesn't affect the asymptotic runtime.} 
    
    One solution is to run binary search on each column of the matrix; this involves looking at $O(k \log k)$ entries of the matrix. The minimum indices will clearly be found this way, as the minimum point will be the lowest $1$ entry in some column. Evaluating each entry requires us to run \autoref{alg:chunkSplit}, which runs in $O(k^2)$. Thus, the total runtime over all edges in the graph is $O(|E|^2 k^3 \log k)$. This brings the total computation cost to $O(|E|^2 k^3 \log k + |V|)$. The correctness of $\mathcal{P}'$ follows obviously from the correctness of \autoref{alg:opt-edge-chunk}, \autoref{alg:chunkSplit}, and \autoref{alg:greedy2}.
\end{proof}

\mAgentsLocal*
\begin{proof}
We simply use \autoref{alg:greedy2} to determine which edges can be chunked such that all agents will take the chunking. Keep only those edges in the graph, and run a shortest-path algorithm. This is exactly analogous to \autoref{thm:local-budget}, except that the runtime increases by a factor of $m$ because \autoref{alg:greedy2} runs in $O(mk)$ time.
\end{proof}
\mAgentsGlobal*
\begin{proof}
This is exactly the same as the proof of \autoref{thm:global-budget}, except that we use binary search to find the minimum number of chunks $l_{e}$ such that \textit{all} agents take the optimal $l_{e}$ chunking of edge $e$. Thus, we run \autoref{alg:greedy2} $\log k$ times for each edge, resulting in a total runtime of $O(|E|mk\log k + |V|)$. 
\end{proof}

\end{document}